\documentclass[lettersize,journal]{IEEEtran}
\usepackage{cite}
\usepackage{amsmath,amssymb,amsthm}
\usepackage{algorithmic}
\usepackage{array}
\usepackage[caption=false,font=normalsize,labelfont=sf,textfont=sf]{subfig}
\usepackage{textcomp}
\usepackage{stfloats}
\usepackage{url}
\usepackage{verbatim}
\usepackage{graphicx}
\usepackage{color}
\usepackage{xcolor}
\usepackage{multirow}
\usepackage{colortbl}
\usepackage{diagbox}
\usepackage{booktabs}
\usepackage{makecell}
\usepackage[hidelinks]{hyperref}
\usepackage{orcidlink}
\usepackage[T1]{fontenc}
\usepackage{mathtools} 
\newtheorem{theorem}{Theorem}
\newtheorem{definition}{Definition}

\newtheorem{type}{Type}
\begin{document}

\title{A Robust Learning Framework for Deep-Learning-Based Radar Target Detection With Partially Mislabeled Training Data}

\author{Chuanfei Zang\orcidlink{0009-0005-8851-5926}, 
	Yiru Lin\orcidlink{0009-0004-6850-5344},
	Yumiao Wang\orcidlink{0000-0002-1474-0486},
	Xingyu Chen\orcidlink{0009-0001-4295-9928},
	Xiaobo Yang\orcidlink{0009-0000-4888-4050},\\
	and Guolong Cui\orcidlink{0000-0001-5707-6311},~\IEEEmembership{Senior Member,~IEEE,}
	\thanks{This work was supported by the National Natural Science Foundation of China under Grants 62271126 and U24B20188, and by the Sichuan Province Innovative Talent Funding Project for Postdoctoral Fellows. (\textit{Corresponding author: Yumiao Wang}.)}

	\thanks{Chuanfei Zang, Yiru Lin, Yumiao Wang, Xingyu Chen, and Xiaobo Yang are with the School of Information and Communication Engineering, University of Electronic Science and Technology of China, Chengdu 611731, China (e-mail: zangcf2017@163.com; linyiru2021@163.com; wangyumiao@uestc.edu.cn; cxytoby@163.com; xiaoboyang1964@gmail.com). }
	
	\thanks{Guolong Cui is with the School of Information and Communication Engineering, University of Electronic Science and Technology of China, Chengdu 611731, China, and also Tianfu Jiangxi Laboratory, Chengdu, Sichuan 641419, China (e-mail: cuiguolong@uestc.edu.cn). }}

\markboth{Journal of \LaTeX\ Class Files,~Vol.~14, No.~8, August~2021}%
{Shell \MakeLowercase{\textit{et al.}}: A Sample Article Using IEEEtran.cls for IEEE Journals}


\maketitle

\begin{abstract}
	Deep-learning-based radar target detection has substantially improved detection performance, but its effectiveness depends critically on the correctness of the training labels assigned to radar echo samples. In practical applications, training data may contain partially mislabeled samples, which can cause detection models to learn biased supervisory information and thereby degrade detection performance. To address this issue, this paper proposes SL-RLF, a robust learning framework for deep-learning-based radar target detection with partially mislabeled training data. Specifically, we establish a probabilistic model to characterize the generation mechanisms of label errors in radar target detection and analyze the robustness of detection models under partial label errors within a risk minimization formulation. Theoretical analysis shows that, under appropriate assumptions, a loss function satisfying the symmetry condition enables the detection model in the learning framework to be robust against partial label errors. Guided by this result, we design a loss function satisfying the symmetry condition and construct the proposed robust learning framework, SL-RLF. Theoretically, under the corresponding label-error conditions, the optimal detection model learned by SL-RLF from partially mislabeled training data can achieve the same detection performance under the true data distribution as the optimal model learned from accurately labeled data. Experimental results demonstrate that, under different types and levels of partial label errors, the proposed SL-RLF consistently outperforms representative methods, including Co-teaching and TCE, and its robustness is generally consistent with the theoretical analysis.
\end{abstract}


\section{Introduction}
Radar target detection is a fundamental task in radar signal processing \cite{10048578,9714856,9854805}. By transmitting electromagnetic waves and receiving target-scattered echoes, radar systems enable long-range target detection and can operate under all-weather and all-time conditions \cite{10700971,9420308}. Consequently, they play an important role in both military and civilian applications\cite{10716783,11124398}. However, in complex environments, target echoes are often severely masked by strong clutter, which substantially degrades target detectability and poses significant challenges to reliable radar target detection \cite{10975280,9457195}.

Conventional radar target detection methods typically rely on statistical clutter modeling. When the assumed clutter distribution matches the observed data, these methods can achieve satisfactory detection performance; however, variations in clutter statistics may cause model mismatch and performance degradation \cite{10929070,9733186}. To reduce reliance on prior distributional assumptions, methods exploiting discriminative features between targets and clutter have been developed \cite{8557319,1573749,5395653}. By extracting amplitude, polarimetric, time-frequency, or spatial-correlation features, these methods improve target--clutter separability in complex backgrounds \cite{9380220,1593337}. Nevertheless, their performance remains limited by feature design, parameter selection, and scene adaptability \cite{10587212,1015011}.

Deep-learning-based radar target detection methods are typically trained within a supervised learning framework. These methods can adaptively learn discriminative representations between targets and clutter from labeled echo data, thereby reducing the reliance on explicit mathematical modeling and feature design \cite{10371351,10877895,10485649}. Within this supervised learning framework, recent studies have exploited the advantages of deep learning models in high-dimensional feature extraction, complex nonlinear mapping, and spatiotemporal correlation mining, leading to improved detection performance in complex scenarios \cite{10848402,10582311,11204961,9672704,9018199}.

Although deep learning has substantially improved radar target detection performance, effective training of detection models within the supervised learning framework relies on the accuracy of the labels assigned to radar echo data \cite{9335266,9210028}. Unlike optical images and natural language data, radar echoes lack intuitive semantic representations and usually require sophisticated signal processing procedures to improve the signal-to-clutter ratio (SCR) and highlight target characteristics \cite{10032174,10963724}. Therefore, radar data annotation depends strongly on professional expertise and domain knowledge. It is not only costly, but also difficult to guarantee complete label correctness at all times \cite{10707336,10883047}. In addition, subjective bias from annotators, interference from complex environments, and sensor anomalies may all introduce label errors. In particular, when targets are submerged in strong clutter or observed under low SCR conditions, even experienced experts may find it difficult to provide consistently stable and reliable labels \cite{9585088}. 

Therefore, practical radar datasets inevitably contain a certain proportion of mislabeled samples. In a supervised learning framework, detection models rely on the supervisory information provided by training labels. When a subset of labels is erroneous, the strong function approximation and data memorization capabilities of deep neural networks may cause the detection model to fit, memorize, and reinforce erroneous supervisory signals during training \cite{pmlr-v119-ma20c}. Consequently, under the standard supervised learning framework, the decision relationships learned by the detection model may deviate from the true data distribution, leading to substantial degradation in detection performance \cite{liu2020early,ghosh2017robust}. These issues indicate the need for a robust learning framework that enables deep-learning-based radar target detection models to learn reliably from partially mislabeled training data while suppressing the performance degradation caused by label errors \cite{NEURIPS2021_fe2d0103}.

Arpit \textit{et al.} \cite{pmlr-v70-arpit17a} found that deep-learning-based models trained on partially mislabeled data under a supervised learning framework often exhibit ``early learning,'' in which the model first fits samples with correct labels and subsequently memorizes mislabeled samples. Motivated by this phenomenon, Liu \textit{et al.} \cite{liu2020early} improved the learning framework by introducing a regularization term in the later stage of training to constrain model outputs to remain consistent with early-stage predictions, thereby mitigating overfitting to mislabeled samples. However, the effectiveness of such early-learning-based robust training methods depends heavily on the reliability of the detection results obtained in the initial training stage. In low-SCR scenarios, the model may fail to form a reliable decision boundary at this stage, and its early detection results can therefore be substantially affected by label errors. Moreover, the weight of the introduced regularization term is sensitive to data characteristics. If improperly specified, it may cause the model to rely excessively on initial predictions or fail to sufficiently suppress the memorization of mislabeled samples in later training.

Also motivated by early learning, Wang \textit{et al.} \cite{10675372} proposed a pseudo-labeling (PL) learning framework. This framework first performs preliminary training on partially mislabeled data to obtain early-stage model representations. It then formulates the identification of mislabeled samples as an anomaly detection problem for data cleaning. Finally, the model is retrained with a pseudo-label assignment strategy to reduce the impact of label errors on detection performance. However, this framework mainly addresses the case in which targets are mislabeled as clutter, while giving insufficient consideration to the also practical case in which clutter is mislabeled as targets. In addition, the multiple training stages of the PL framework are mutually dependent, and each stage introduces additional hyperparameters, which significantly increases training cost and tuning complexity.

In contrast to the aforementioned approaches, which improve the robustness of the learning framework by introducing multiple training stages, Han \textit{et al.} \cite{han2018co} proposed the Co-teaching learning framework. This framework simultaneously trains two detection models with identical architectures but different parameter initializations, enabling them to produce slightly different outputs for the same samples and thereby provide discriminative information from distinct perspectives. Based on this mechanism, the framework filters the training data according to the losses and prediction consistencies of both models, retaining samples that are more likely to be correctly labeled to mitigate the interference of mislabeled samples. Although this framework avoids complicated multistage training procedures, it requires training an additional model of the same scale, thereby incurring relatively high computational costs. Furthermore, as training progresses, the representations learned by the two models may gradually converge, reducing the complementarity between the two views and weakening the effectiveness of sample selection in suppressing label errors.

In addition to modifying training strategies, loss function design provides an effective way to construct robust learning frameworks. Supervised training pipelines commonly adopt cross-entropy (CE) as the loss function; however, Lei \textit{et al.} \cite{10.5555/3491440.3491745} showed that CE may cause detection models to overfit mislabeled samples during training. They therefore incorporated the Taylor cross-entropy (TCE) loss into their learning framework. TCE performs a Taylor expansion of the CE loss and truncates higher-order terms, thereby controlling the extent to which the model fits mislabeled samples. However, the performance of this framework depends critically on the truncation order, which typically must be manually adjusted according to the severity of label errors. This dependence limits the generality of the framework across different radar target detection scenarios.

Although existing methods can mitigate the adverse effects of label errors on model training to some extent, they often involve complex training procedures, additional hyperparameters, or limited theoretical interpretability. This paper develops a theoretically grounded robust learning framework that does not introduce additional training stages, sample-selection strategies, or hyperparameters. Specifically, we first model the generation mechanisms of label errors in radar target detection and analyze, within a Bayes risk minimization formulation, the conditions under which loss functions in a supervised learning framework remain robust to partial label errors. Guided by the derived symmetry condition, we design a symmetric loss to replace the CE loss in the standard supervised learning framework, thereby obtaining the proposed symmetric-loss-based robust learning framework (SL-RLF). The main contributions of this paper are summarized as follows:
\begin{enumerate}
	\item We probabilistically model the generation process of label errors in radar target detection, establishing a theoretical basis for analyzing the robustness of detection models trained with partially mislabeled data.
	\item Within a Bayes risk minimization formulation, we prove that loss functions satisfying the symmetry condition are robust under the considered partial-label-error conditions.
	\item We design a loss function satisfying the symmetry condition and use it to construct SL-RLF for deep-learning-based radar target detection models. Under appropriate conditions, the optimal detection model learned within SL-RLF from partially mislabeled training data can achieve the same detection performance as that learned from accurately labeled data.  
	\item We conduct experiments on measured radar echo data to validate SL-RLF. The results show that, across different types and levels of label errors, SL-RLF consistently outperforms representative robust-learning methods, including Co-teaching and TCE, and the observed performance trends are highly consistent with the theoretical analysis.
\end{enumerate}

The remainder of this paper is organized as follows. Section \ref{section:Signal_Model} presents the signal model and problem formulation. Section \ref{section:Methodology} presents the theoretical analysis and related theorems. Section \ref{section:Experiments} provides experimental evaluation and comparative analysis. Finally, Section \ref{section:Conclusion} concludes the paper.

\section{Signal Model and Learning Framework} \label{section:Signal_Model}
During a coherent processing interval (CPI), the radar transmits $P$ coherent pulses and receives the corresponding echoes through the I/Q channels. After quadrature demodulation, matched filtering, and range extraction, a discrete echo matrix $\mathbf{X} \in \mathbb{C}^{P \times R}$ is obtained, where $R$ denotes the number of range cells, and $\mathbf{x}_r \in \mathbb{C}^{P \times 1}$ denotes the slow-time echo vector corresponding to the $r$-th range cell. For an arbitrary cell under test (CUT), the received echo can be described by the following binary hypothesis testing model:
\begin{equation}
	\begin{cases}
		H_0:\mathbf{x}_r=\mathbf{c}_r+\mathbf{n}_r,\\
		H_1:\mathbf{x}_r=\mathbf{s}_r+\mathbf{c}_r+\mathbf{n}_r,
	\end{cases}
	\label{eq:target}
\end{equation}
where $\mathbf{s}_r$, $\mathbf{c}_r$, and $\mathbf{n}_r$ denote the target-scattered signal, clutter signal, and additive noise, respectively. The hypotheses $H_1$ and $H_0$ correspond to target-present and target-absent CUTs, respectively; these two types of CUTs are hereafter referred to as target cells and clutter cells.

Based on the hypothesis testing model, conventional detection methods aim to make a binary decision between $H_0$ and $H_1$ given $\mathbf{x}_r$, without further identifying the specific class of the target contained in a target cell. In contrast, deep-learning-based detection methods for radar targets can formulate radar target detection as a multiclass classification problem \cite{9999169}. Specifically, clutter cells and target cells containing different target classes are treated as mutually exclusive classes, and the model outputs the predicted probability that a sample $\mathbf{x}_r$ belongs to each class. The detection decision is then made according to these probabilities. For notational simplicity, the range-cell index $r$ is omitted hereafter, and $\mathbf{x}$ denotes the echo signal of the current CUT.

To get the network parameters of the detection model, the model is trained within a supervised learning framework using training data with correct labels. Let the training dataset be denoted by $\mathcal{D} = \{(\mathbf{x}_i, y_{\mathbf{x}_i})\}_{i=1}^N$, sampled from an unknown scene distribution over $\mathcal{P}=\mathcal{X} \times \mathcal{Y}$, where $i$ is the sample index, $N$ is the number of samples, and $y_{\mathbf{x}_i}$ is the label of sample $\mathbf{x}_i$. Here, $\mathcal{X} \subseteq \mathbb{C}^{P \times 1}$ denotes the complex-valued signal space, and $\mathcal{Y} = \{1, \dots, K\}$ denotes the label space, where $K$ is the total number of classes.

Within this framework, the detection model $f_{\boldsymbol{\theta}}$ maps the input signal to the $K$-dimensional probability simplex $\Delta^{K-1}$:
\begin{equation}
	f_{\boldsymbol{\theta}}\!:\!\mathcal{X}\!\to\!\Delta^{K\!-\!1}\!,\
	\Delta^{K\!-\!1}\!\!=\!\left\{\mathbf{p}\!\in\!\mathbb{R}^{K}\middle| \mathbf{p}_k\ge 0, \sum_{k=1}^{K}\mathbf{p}_k \!=\!1\right\}.
	\label{eq:model_output}
\end{equation}
Thus, for any input sample $\mathbf{x}_i$, the model output $f_{\boldsymbol{\theta}}(\mathbf{x}_i)$ is a $K$-dimensional vector whose entries are nonnegative and sum to one:
\begin{equation}
	\sum_{k=1}^{K}[f_{\boldsymbol{\theta}}(\mathbf{x}_i)]_k=1, \quad[f_{\boldsymbol{\theta}}(\mathbf{x}_i)]_k \ge 0,
\end{equation}
where  $[f_{\boldsymbol{\theta}}(\mathbf{x}_i)]_k$ denotes the $k$-th component of $f_{\boldsymbol{\theta}}(\mathbf{x}_i)$ , representing the predicted probability that the input sample $\mathbf{x}_i$ belongs to the $k$-th class. Based on this output, the decision function $\mathrm{dec}\left( \cdot\right) $ in this framework adopts the maximum a posteriori criterion to obtain the final detection result:
\begin{equation}
	\mathrm{dec}\!\left(f_{\boldsymbol{\theta}}(\mathbf{x}_i)\right)=\arg\max_{k}[f_{\boldsymbol{\theta}}(\mathbf{x}_i)]_k.
\end{equation}

The corresponding loss function $\ell: \Delta^{k-1} \times \mathcal{Y} \to \mathbb{R}_+$ measures the consistency between the model output $f_{\boldsymbol{\theta}} (\mathbf{x}_i)$ and the label $y_{\mathbf{x}_i}$ represented in one-hot form. Given $\ell$, the population risk of the model $f_{\boldsymbol{\theta}}$ is defined as
\begin{equation}
	\begin{aligned}
		R_\ell(f_{\boldsymbol{\theta}}) = \mathbb{E}_{\mathbf{x}_i, y_{\mathbf{x}_i}}[\ell\left( f_{\boldsymbol{\theta}}(\mathbf{x}_i), \mathbf{e}_{y_{\mathbf{x}_i}}\right) ],
	\end{aligned}
	\label{eq:R_L}
\end{equation}
where $\mathbf{e}_i$ denotes the unit vector whose $i$-th component is one and all other components are zero. During model training, the parameters of $f_{\boldsymbol{\theta}}$ are updated by minimizing $R_\ell(f_{\boldsymbol{\theta}})$ \cite{10.5555/944919.944944}, and the optimal solution is given by
\begin{equation}
	f_{{\boldsymbol{\theta}}^\ast} = \arg\min_{\boldsymbol{\theta}}  {R}_\ell(f_{\boldsymbol{\theta}}).
	\label{eq:Correct_data_Optimal}
\end{equation}
Within this learning framework, the mean squared error (MSE) loss and CE loss have been widely used to train  $f_{\boldsymbol{\theta}}$ because of their stable optimization behavior and favorable convergence performance \cite{1227801}.

However, this supervised learning framework relies on the correctness of the labels in the training dataset $\mathcal{D}$ to learn a reliable detection model. This requirement is difficult to satisfy rigorously in practice, because the labeling process may be affected by human judgment bias, complex environmental interference, sensor anomalies, and other factors, causing the observed labels of some samples to deviate from their true classes. Consequently, real radar datasets may contain a certain proportion of mislabeled samples. In this case, if the detection model is still trained under this learning framework, the strong fitting and memorization capabilities of deep neural networks may cause the model to overlearn erroneous supervisory information from mislabeled samples. As a result, the model may learn decision relationships that deviate from the true data distribution, ultimately degrading detection performance \cite{karimi2020deep}. Therefore, developing a robust learning framework that trains detection models on partially mislabeled training data while suppressing the adverse impact of label errors on detection performance is the central problem addressed in this paper.

\section{Methodology} \label{section:Methodology}
This section first establishes a probabilistic model to characterize the generation mechanisms of label errors. Based on this model, the robustness of loss functions satisfying the symmetry condition is analyzed under different label-error generation mechanisms within a risk minimization formulation. Guided by the symmetry condition, a symmetric loss function is then designed, and a robust learning framework with theoretical robustness guarantees is obtained for training deep-learning-based radar target detection models on partially mislabeled training data by replacing the CE loss in the standard supervised learning framework with the proposed symmetric loss.

\begin{figure}[t]
	\centering
	\includegraphics[width=0.49\textwidth, height=0.21\textwidth]{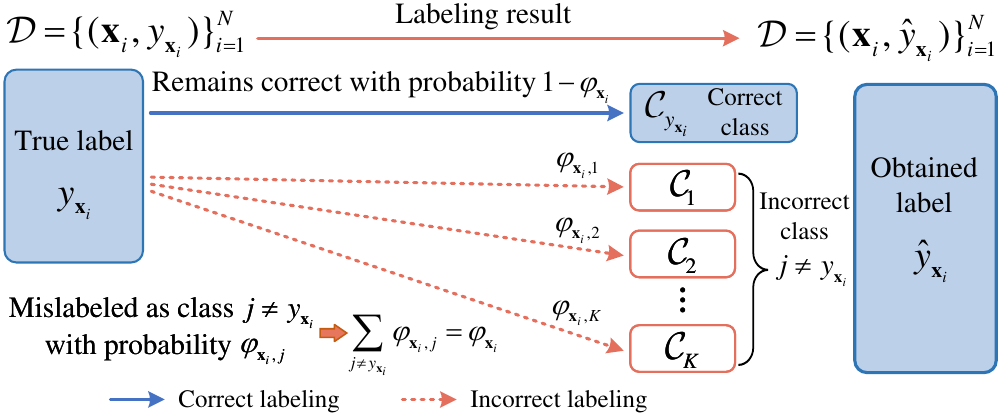}
	\caption{Schematic diagram of the incorrect label generation process.}
	\label{fig:gen_error_label}
\end{figure}

\subsection{Label-Error Modeling}\label{Error type modeling}
As illustrated in Fig.~\ref{fig:gen_error_label}, let the partially mislabeled training set be denoted by $\hat{\mathcal{D}} = \{(\mathbf{x}_i, \hat{y}_{\mathbf{x}_i})\}_{i=1}^N$, and let the corresponding accurately labeled dataset be denoted by $\mathcal{D} = \{(\mathbf{x}_i, y_{\mathbf{x}_i})\}_{i=1}^N$. Here, $y_{\mathbf{x}_i}$ denotes the true label of $\mathbf{x}_i$, whereas $\hat{y}_{\mathbf{x}_i}$ denotes the observed label obtained through the labeling process, which may differ from $y_{\mathbf{x}_i}$. The label-error generation process is modeled as
\begin{equation}
	\hat{y}_{\mathbf{x}_i} =
	\begin{cases}
		y_{\mathbf{x}_i} & {\mathrm{~w.p.~}} 1 - \varphi_{\mathbf{x}_i} ;\\
		j \in \{1,...,K\} \setminus \{y_{\mathbf{x}_i}\} & {\mathrm{~w.p.~}} {\varphi}_{\mathbf{x}_i, j},
	\end{cases}
	\label{eq:get_label}
\end{equation}
where $\varphi_{\mathbf{x}_i}$ denotes the total probability that sample $\mathbf{x}_i$ is mislabeled, referred to as the label-error probability. Moreover, ${\varphi}_{\mathbf{x}_i, j}$ denotes the transition probability that sample $\mathbf{x}_i$ is mislabeled as class $j$ instead of its true label $y_{\mathbf{x}_i}$, satisfying
\begin{equation}
	\sum_{j \neq y_{\mathbf{x}_i}} {\varphi}_{\mathbf{x}_i, j} = \varphi_{\mathbf{x}_i}.
\end{equation}

Equation~\eqref{eq:get_label} characterizes the generation mechanism from the true label $y_{\mathbf{x}_i}$ to the actually observed label $\hat{y}_{\mathbf{x}_i}$. According to the dependence of $\varphi_{\mathbf{x}_i}$ and ${\varphi}_{ \mathbf{x}_i, j}$ on the sample $\mathbf{x}_i$ and its true class $y_{\mathbf{x}_i}$, the label errors in $\hat{\mathcal{D}}$ are categorized into the following three types, as shown in Fig.~\ref{fig:three_type_error}. Only these three types are considered in the subsequent theoretical analysis.

\begin{type} \label{type 1}
	{\textbf{Constant uniform error:}} The label-error probability $\varphi_{\mathbf{x}_i}$ is the same constant $\varphi_0$ for all samples. Moreover, once the label $y_{\mathbf{x}_i}$ is incorrect, $\hat{y}_{\mathbf{x}_i}$ is equally likely to be assigned to any of the remaining $K-1$ classes:
	\begin{equation}
		\varphi_{\mathbf{x}_i}=\varphi_0, \quad {\varphi}_{\mathbf{x}_i, j} = \frac{\varphi_0}{K-1}, \quad 	\forall j \neq y_{\mathbf{x}_i}.
	\end{equation}
\end{type}
This model corresponds to random errors that are independent of both sample features and true classes.

\begin{type} \label{type 2}
	{\textbf{Sample-dependent uniform error:}} The label-error probability $\varphi_{\mathbf{x}_i}$ is a function of the sample $\mathbf{x}_i$, whereas ${\varphi}_{\mathbf{x}_i, j}$ remains uniformly distributed over all incorrect classes:
	\begin{equation}
		\varphi_{\mathbf{x}_i}=g\left( \mathbf{x}_i\right) \in\left[ 0,1\right], \quad {\varphi}_{\mathbf{x}_i, j} = \frac{\varphi_{\mathbf{x}_i}}{K-1}, \quad \forall j \neq y_{\mathbf{x}_i}.
	\end{equation}
\end{type}
This model describes cases in which the label-error probability varies with sample difficulty. For example, samples with low SCR or weak target charact eristics are more likely to be mislabeled.

\begin{type} \label{type 3}
	{\textbf{Class-dependent error:}} Both $\varphi_{\mathbf{x}_i}$ and ${\varphi}_{\mathbf{x}_i,j}$ depend only on the true class $y_{\mathbf{x}_i}$ of the sample, rather than on the specific sample $\mathbf{x}_i$:
	\begin{equation}
		\varphi_{\mathbf{x}_i} = \varphi_{y_{\mathbf{x}_i}}, \quad {\varphi}_{\mathbf{x}_i, j} = {\varphi}_{y_{\mathbf{x}_i}, j}.
	\end{equation}
\end{type}
This model capture class-dependent differences in label-confusion patterns. For instance, in radar echoes, small vessels may be more likely to be mislabeled as buoys, whereas large vessels rarely exhibit such errors.

\subsection{Theoretical Analysis of Robustness}\label{math}
From \eqref{eq:Correct_data_Optimal}, when the model $f_{\boldsymbol{\theta}}$ is trained on the accurately labeled dataset $\mathcal{D}$, its optimal solution in the sense of risk minimization is $f_{\boldsymbol{\theta}^\ast}$. When the model is trained on the partially mislabeled dataset $\hat{\mathcal{D}}$, the corresponding population risk is given by
\begin{equation}
	\hat{R}_\ell(f_{\boldsymbol{\theta}}) = \mathbb{E}_{\mathbf{x}_i, \hat{y}_{\mathbf{x}_i}}[\ell(f_{\boldsymbol{\theta}}(\mathbf{x}_i), \mathbf{e}_{\hat{y}_{\mathbf{x}_i}})],
\end{equation}
and its optimal solution is denoted by
\begin{equation}
	\hat{f}_{\boldsymbol{\theta}^\ast} = \arg\min_{\boldsymbol{\theta}}\hat{R}_\ell(f_{\boldsymbol{\theta}}).
\end{equation}

As a basis for constructing a robust learning framework for training detection models with partially mislabeled data, we define the robustness of the training loss as follows.

\begin{definition} \label{definition 1}
	For a given loss function $\ell$, if the optimal model $f_{\boldsymbol{\theta}^{\ast}}$ trained with accurate labels and the optimal model $\hat{f}_{\boldsymbol{\theta}^{\ast}}$ trained with partially mislabeled data achieve the same detection performance under the true data distribution $\mathcal{P}$, then $\ell$ is said to render the detection model robust to the label errors in $\hat{\mathcal{D}}$, i.e.,
	\begin{equation}
		\hspace{-0.2cm}
		\Pr_{\left( \mathbf{x},y_{{\mathbf{x}}}\right) \sim \mathcal{P}} \!\left[ {\mathrm{dec}}\!\left(f_{\boldsymbol{\theta}^\ast}\!\left( \mathbf{x}\right)\right) \!=\! y_{{\mathbf{x}}}\right]  \!= \!\!\!\!\!\Pr_{\left( \mathbf{x},y_{{\mathbf{x}}}\right) \sim \mathcal{P}}\!\left[ \mathrm{dec}\! \left(  \hat{f}_{\boldsymbol{\theta}^\ast}\!\left( \mathbf{x}\right) \right)  \!\!=\! y_{{\mathbf{x}}}\right]\!.\!\!
		\label{eq:equal_detect}
	\end{equation}
\end{definition}

In this case, if $ f_{\boldsymbol{\theta}^\ast} $ is also an optimal solution of the population risk $\hat{R}_\ell$ on $\hat{\mathcal{D}}$, i.e., 
\begin{equation}
	\hat{R}_\ell(f_{\boldsymbol{\theta}^\ast}) - \hat{R}_\ell(f_{\boldsymbol{\theta}})  \leq 0, \quad \forall {\boldsymbol{\theta}}.
\end{equation} 
then can take $\hat{f}_{\boldsymbol{\theta}^\ast} = f_{\boldsymbol{\theta}^\ast}$, and Definition~\ref{definition 1} is consequently satisfied.

\begin{definition} \label{definition 2}
	If there exists a constant $C$ such that, for any input sample $\mathbf{x}_i$ and any model output $f_{\boldsymbol{\theta}}(\mathbf{x}_i)$,
	\begin{equation}
		\ell(f_{\boldsymbol{\theta}}(\mathbf{x}_i), \mathbf{e}_k) \ge 0, \quad \sum_{k=1}^{K} \ell(f_{\boldsymbol{\theta}}(\mathbf{x}_i), \mathbf{e}_k) = C .
		\label{cont}
	\end{equation}
	then $\ell$ is said to be a symmetric loss function.
\end{definition}

Based on these definitions, we theoretically derive sufficient conditions under which symmetric loss functions render the detection model robust against partial label errors for the three label-error types defined above.

\begin{figure*}[t]
	\centering
	\includegraphics[width=0.99\textwidth, height=0.44\textwidth]{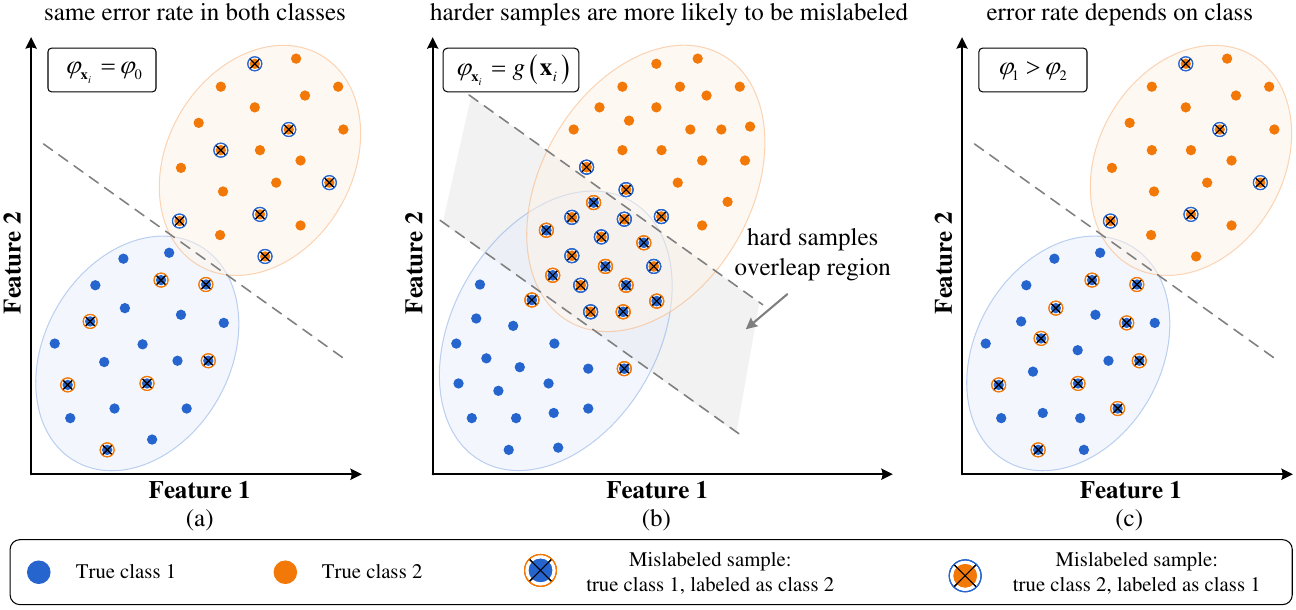}
	\vspace{-0.2cm}
	\caption{Illustration of the three representative types of label error mechanisms modeled in this pape: (a) Type~1, constant uniform error; (b) Type~2, sample-dependent uniform error; (c) Type~3, class-dependent error.}
	\label{fig:three_type_error}
	\vspace{-0.4cm}
\end{figure*}

\begin{theorem} \label{theorem 1}
	Let $\ell$ be a symmetric loss function. Under the Type~1 error model, i.e., constant uniform error, if
	\begin{equation}
		\varphi_0 < \frac{K-1}{K},
	\end{equation}
	then $\ell$ renders the detection model robust against training data affected by this type of label error. The proof is provided in Appendix~1.
\end{theorem}

\begin{theorem} \label{theorem 2}
	Let $\ell$ be a symmetric loss function. Under the Type~2 error model, i.e., sample-dependent uniform error, if, for any sample $\mathbf{x}_i \in \mathcal{X}$,
	\begin{equation}
		\varphi_{\mathbf{x}_i} < \frac{K-1}{K},
	\end{equation}
	and $R_\ell(f_{\boldsymbol{\theta}^\ast}) = 0$, then $\ell$ renders the detection model robust against training data affected by this type of label error. The proof is provided in Appendix~2.
\end{theorem}

\begin{theorem} \label{theorem 3}
	Let $\ell$ be a symmetric loss function. Under the Type~3 error model, i.e., class-dependent error, suppose that $\ell$ further satisfies, for any class $k$ and any sample $\mathbf{x}_i \in \mathcal{X}$,
	\begin{equation}
		0 \leq \ell\big(f_{\boldsymbol{\theta}}(\mathbf{x}_i), \mathbf{e}_k\big) \leq \frac{C}{K-1},
		\label{eq:theorem_3}
	\end{equation}
	and that the optimal risk under accurate labels satisfies $R_\ell(f_{\boldsymbol{\theta}^\ast}) = 0$. If, for any true label $y_{\mathbf{x}_i}$ and any incorrect class $j \neq y_{\mathbf{x}_i}$,
	\begin{equation}
		\varphi_{y_{\mathbf{x}_i},j} < 1 - \varphi_{y_{\mathbf{x}_i}} = 1- \sum_{j \neq y_{\mathbf{x}_i}} {\varphi}_{y_{\mathbf{x}_i}, j},
		\label{eq:error-probability-condition}
	\end{equation}
	then $\ell$ renders the detection model robust against training data affected by this type of label error. The proof is provided in Appendix~3. 
\end{theorem}

This error-probability condition of \eqref{eq:error-probability-condition} in Theorem~\ref{theorem 3} implies that the corresponding label transition matrix is strictly diagonally dominant. Specifically, for each true class, the probability that its label remains correct is strictly greater than the probability that it is transferred to any  incorrect class.

These theorems show that symmetry is a key structural condition for a loss function to achieve robustness in the presence of partial label errors, and the conclusions of all theorems are independent of the specific form of the data distribution. In comparison, Theorem~\ref{theorem 1} has the mildest applicability condition, requiring only that the error probability satisfy $\varphi_0<(K-1)/K$. By contrast, Theorems~\ref{theorem 2} and~\ref{theorem 3} additionally require that the optimal risk under accurate labels satisfy $R_\ell(f_{\boldsymbol{\theta}^\ast})=0$. This condition requires the deep-learning-based detection model to have sufficient representational capacity to fit the training data adequately. For deep neural networks, this requirement is not difficult to meet. In addition, Theorem~\ref{theorem 3} requires the symmetric loss function to satisfy a boundedness condition; therefore, its scope of applicability is comparatively more restricted.

\subsection{Framework Design}\label{Loss Design}
The preceding theoretical analysis shows that, for a loss function used to train the detection model, satisfying the symmetry condition in Definition~\ref{definition 2} is essential for establishing theoretical robustness guarantees against label errors in partially mislabeled training data. We therefore examine the structural properties of loss functions commonly used in supervised learning frameworks. The CE loss and MSE loss are defined as
\begin{align}
	\!\!&\ell_{\text{CE}}(f_{\boldsymbol{\theta}}(\mathbf{x}_i), \mathbf{e}_j) \!=\! -\!\!\sum_{k=1}^{K}\! \mathbf{e}_j^k \log\! {\left[ f_{\boldsymbol{\theta}}(\mathbf{x}_i)\right]_k} \!=\! -\!\log\! {\left[ f_{\boldsymbol{\theta}}(\mathbf{x}_i)\right] _j},\!\\
	&\ell_{\text{MSE}}(f_{\boldsymbol{\theta}}(\mathbf{x}_i), \!\mathbf{e}_j) \!\!=\! \left\| \mathbf{e}_j \!-\!\! f_{\boldsymbol{\theta}}(\mathbf{x}_i)\right\|_2^2 \!\!=\!\! ||f_{\boldsymbol{\theta}}(\mathbf{x}_i)||_2^2 \!\!+\!\! 1 \!\!-\!\! 2\left[ f_{\boldsymbol{\theta}}(\mathbf{x}_i)\right]_j, \nonumber
\end{align}
where $\mathbf{e}_j^i$ denotes the $i$-th component of $\mathbf{e}_j$. Summing the losses over all classes yields
\begin{equation}
	\begin{aligned}
		&\quad\sum_{k=1}^{K}\ell_{\text{CE}}(f_{\boldsymbol{\theta}}(\mathbf{x}_i), \mathbf{e}_k) =- \sum_{k=1}^{K} \log {\left[ f_{\boldsymbol{\theta}}(\mathbf{x}_i)\right] _k},\\
		&\sum_{k=1}^{K}\ell_{\text{MSE}}(f_{\boldsymbol{\theta}}(\mathbf{x}_i), \mathbf{e}_k) = K ||f_{\boldsymbol{\theta}}(\mathbf{x}_i)||_2^2 + K - 2.
	\end{aligned}
\end{equation}
Thus, for both the CE and MSE losses, the sum of losses over all classes depends on the model output $f_{\boldsymbol{\theta}}(\mathbf{x}_i)$. Consequently, neither loss satisfies the symmetry condition. Therefore, the robustness results established above for symmetric loss functions cannot be directly applied to the CE or MSE loss commonly used in standard supervised learning frameworks. This analysis indicates that, to provide theoretical robustness guarantees for the learning framework in the presence of partially mislabeled training data, it is necessary to construct a loss function that satisfies the symmetry condition.

Motivated by this, we design the following robust loss function:
\begin{equation}
	\ell_{\text{ROB}}\left( f_{\boldsymbol{\theta}}(\mathbf{x}_i), \mathbf{e}_k\right) = 1 - \left[ f_{\boldsymbol{\theta}}(\mathbf{x}_i)\right] _k, \ \forall k \in \left\lbrace1,\cdots,K \right\rbrace .
	\label{eq:we_loss}
\end{equation}
This loss depends only on the output component corresponding to the labeled class and requires no additional hyperparameters. Since the model output satisfies the probability normalization condition $\sum_{k=1}^{K}f_{\boldsymbol{\theta}}(\mathbf{x}_i)_k=1$, we have
\begin{equation}
	\sum_{k=1}^{K}\ell_{\text{ROB}}(f_{\boldsymbol{\theta}}(\mathbf{x}_i), \mathbf{e}_k) = K - 1.
\end{equation}
This verifies that $\ell_{\mathrm{ROB}}$ satisfies the symmetry condition with $C=K-1$.

We further verify whether $\ell_{\mathrm{ROB}}$ satisfies the remaining conditions in the preceding theorems. First, since $\ell_{\text{ROB}}$ is symmetric loss, Theorem~\ref{theorem 1} can be directly applied when the label-error probability satisfies the required constraint. Second, substituting $C=K-1$ into the boundedness condition in Theorem~\ref{theorem 3} gives
\begin{equation}
	0 \leq \ell(f_{\boldsymbol{\theta}}\left( \mathbf{x}_i\right), \mathbf{e}_k) \leq \frac{C}{K-1} = 1, \;\forall k \in \left\lbrace1,\cdots,K \right\rbrace .
\end{equation}
According to \eqref{eq:we_loss} and the nonnegativity of the detection model output, $\ell_{\mathrm{ROB}}$ clearly satisfies this boundedness condition. For the zero-risk condition $R_\ell(f_{\boldsymbol{\theta}^\ast}) = 0$ in Theorems~\ref{theorem 2} and~\ref{theorem 3}, if the deep-learning-based detection model has sufficient data-fitting capacity such that $f_{\boldsymbol{\theta}^\ast}$ can output detection results consistent with the true labels for almost every sample, i.e.,
\begin{equation}
	\left[ f_{\boldsymbol{\theta}^\ast}(\mathbf{x}_i)\right] _{y_{\mathbf{x}_i}}=1 \quad \mathrm{a.s.}\;(\mathbf{x}_i,y_{\mathbf{x}_i})\sim\mathcal{P},
\end{equation}
then, from \eqref{eq:we_loss}, we obtain
\begin{equation}
	\ell_{\mathrm{ROB}} \left( f_{\boldsymbol{\theta}^{\ast}}(\mathbf{x}_i),\mathbf{e}_{y_{\mathbf{x}_i}} \right)
	= 1-[f_{\boldsymbol{\theta}^{\ast}}(\mathbf{x}_i)]_{y_{\mathbf{x}_i}} 	= 0 .
\end{equation}
Together with \eqref{eq:R_L}, this implies that $R_\ell(f_{\boldsymbol{\theta}^\ast}) = 0$.

\begin{figure*}[t]
	\centering
	\includegraphics[width=1.0\textwidth, height=0.22\textwidth]{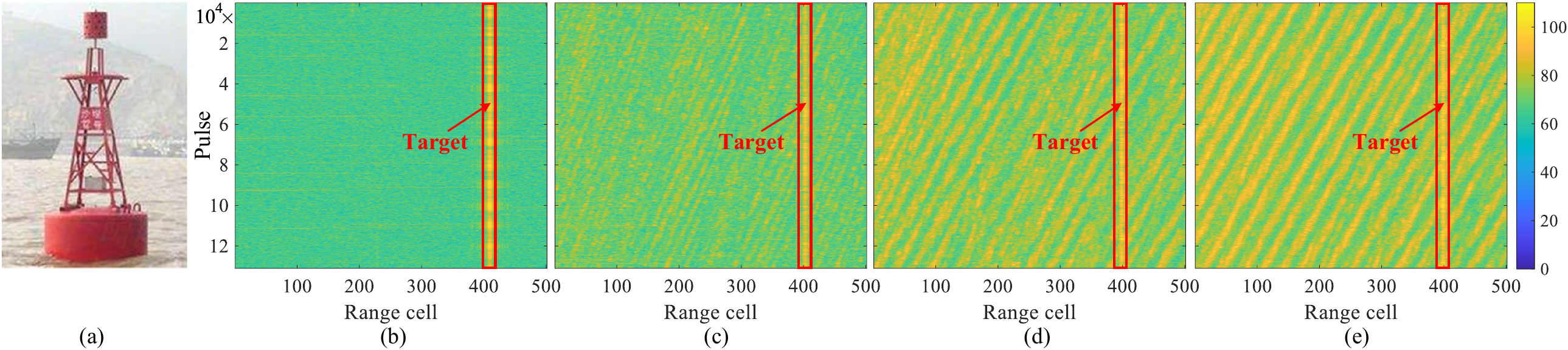}
	\vspace{-0.6cm}
	\caption{MTDSP dataset. (a) Observed buoy target. (b)-(e) Radar echo signals under sea states 2--5. As the sea state increases, the interference from sea clutter becomes progressively stronger.}
	\label{fig:MTDSP}
	\vspace{-0.4cm}
\end{figure*}

Consequently, when the label-error probabilities satisfy the conditions specified in the corresponding theorems, $\ell_{\text{ROB}}$ is theoretically robust under all three considered types of partial label errors. Therefore, replacing the loss function in the standard supervised learning framework with $\ell_{\text{ROB}}$ yields the proposed robust learning framework, SL-RLF, for deep-learning-based radar target detection with partially mislabeled training data. Moreover, because $\ell_{\text{ROB}}$ requires no additional hyperparameters, SL-RLF follows the same training workflow as standard supervised learning with accurately labeled training data. Accordingly, under SL-RLF, the detection model can be trained without modifying its network architecture or training pipeline, thereby mitigating the adverse impact of partial label errors on model parameter learning.

\begin{table}[t]
	\renewcommand{\arraystretch}{1.3}
	\caption{Radar Parameters of MTDSP Database}
	\label{MTDSP Radar parameters}
	\centering
	\setlength{\tabcolsep}{2.2mm}{ 
		\vspace{-0.2cm}
		\begin{tabular}{cccc}
			\bottomrule
			\cmidrule{1-4} 
			{\textbf{Band Width}} & \textbf{Carrier Frequency} & \textbf{PRF} & \textbf{Sample Frequency} \\  
			\cmidrule{1-4} 
			25 MHz & 9.4 GHz & 0.5 ms & 60 MHz\\
			\bottomrule
			\cmidrule{1-4}
	\end{tabular}}
	\vspace{-0.6cm}
\end{table}

\section{Experiments and Results} \label{section:Experiments}
This section first describes the experimental setup, including the datasets, the architecture of the deep-learning-based detection model, training and implementation details, and performance evaluation metrics. On this basis, the proposed SL-RLF is compared with representative methods, including CE~\cite{lecun2015deep}, MSE~\cite{dohare2024loss}, Co-teaching~\cite{han2018co}, and TCE~\cite{10.5555/3491440.3491745}, to evaluate its robustness under different error types and error intensities. Subsequently, ablation experiments are conducted to analyze the effects of the number of target classes, network architecture of the detection model, sea-state level, and polarization mode on the performance of SL-RLF, and to examine the consistency between the experimental results and the theoretical analysis.

\begin{figure*}[t]
	\centering
	\includegraphics[width=1.0\textwidth, height=0.11\textwidth]{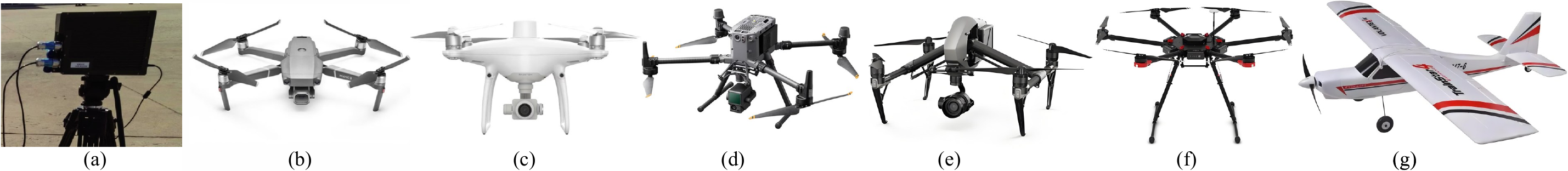}	
	\vspace{-0.6cm}
	\caption{LSS-FMCWR dataset. (a) Ku-band FMCW radar. (b) DJI Mavic 2; (c) DJI Elf; (d) DJI M350; (e) DJI Wu 2; (f) DJI M600; (g) fixed-wing UAV.}
	\label{fig:LSS_FMCWR}
	\vspace{-0.5cm}
\end{figure*}

\subsection{Dataset Description}
We use two publicly available radar datasets in the experiments: the Maritime Target Detection Shared Program (MTDSP) dataset~\cite{jian2023sea} and the Frequency-Modulated Continuous-Wave Radar for Low, Slow, and Small targets (LSS-FMCWR) dataset~\cite{xiaolong2023multiband}.

{\bf{MTDSP:}} This dataset was collected by the SPPR50P X-band solid-state coherent radar deployed in Yantai, Shandong Province, China. Its main operating parameters are listed in Table~\ref{MTDSP Radar parameters}. The radar collected echoes under both HH and VV polarization modes. The reference target was a steel buoy located 2.97 nautical miles from the radar. The data were collected under sea states ranging from 2 to 5. As the sea state increases, the sea-clutter background generally becomes more complex and the interference intensity increases accordingly. Therefore, this dataset is suitable for evaluating the detection performance of different methods under varying sea-clutter intensities. The observed target and radar echo signals are shown in Fig.~\ref{fig:MTDSP}.

{\bf{LSS-FMCWR:}} This dataset was collected using a Ku-band FMCW radar and contains data from six classes of unmanned aerial vehicle (UAV) targets. The main operating parameters of the radar during data acquisition are listed in Table~\ref{LSS-FMCWR Radar parameters}.  Fig.~\ref{fig:LSS_FMCWR} shows the radar prototype and the observed targets. Since this dataset contains multiple target classes, it can be used to analyze the influence of the number of target classes on the robustness of the proposed method.

\begin{table}[t]
	\renewcommand{\arraystretch}{1.3}
	\caption{Radar Parameters of LSS-FMCWR Database}
	\label{LSS-FMCWR Radar parameters}
	\centering
	\setlength{\tabcolsep}{2.0mm}{ 
		\vspace{-0.2cm}
		\begin{tabular}{cccc}
			\bottomrule
			\cmidrule{1-4} 
			{\textbf{Band Width}} & \textbf{Carrier Frequency} & \textbf{PRF} & \textbf{Sample Frequency} \\  
			\cmidrule{1-4} 
			100 MHz & 23.7 GHz & 0.3 ms & 500 KHz\\
			\bottomrule
			\cmidrule{1-4}
	\end{tabular}}
	\vspace{-0.6cm}
\end{table}

\subsection{Model Architecture} \label{section:Model_Architecture}
To verify the generality of the proposed SL-RLF with respect to the network architecture of the detection model, this paper adopts two representative convolutional neural network architectures, ResNet~\cite{he2016deep} and VGG~\cite{simonyan2015very}, for radar target detection. Fig.~\ref{resnet18} illustrates the detailed network structure of ResNet and its basic workflow for radar target detection. ResNet enables cross-layer information propagation through residual connections, which helps alleviate the vanishing-gradient problem in deep network training~\cite{9143512}. Owing to its strong feature representation capability and stable training behavior, ResNet is used as the primary detection model in most performance comparison experiments in this paper. In contrast, as shown in Fig.~\ref{fig:vgg}, VGG adopts a sequentially stacked convolutional structure without cross-layer residual connections, in which hierarchical features are progressively extracted using multiple small convolutional kernels.  Therefore, VGG and ResNet differ substantially in network connectivity, information propagation paths, and feature extraction structures. Based on these structural differences, ResNet and VGG are used as detection models in the ablation experiments to evaluate the architectural generalizability of SL-RLF under partially mislabeled training data.

\begin{figure}[t]
	\centering
	\includegraphics[width=0.45\textwidth, height=0.55\textwidth]{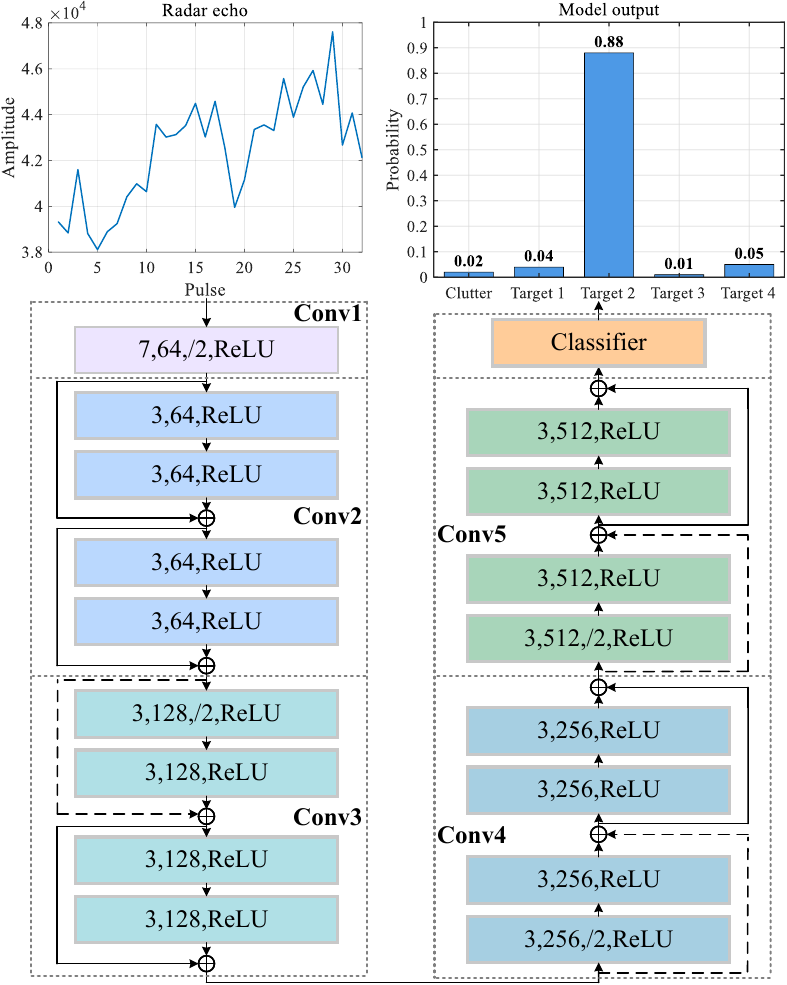}
	\caption{Architecture of the ResNet-based detection model. $7,64,/2,\text{ReLU}$ indicates that the corresponding convolutional module uses a kernel size of 7, 64 output channels, a stride of 2, and the ReLU activation function. Each convolutional layer is followed by batch normalization. It should be noted that the radar signal amplitude map shown in the figure is used only to illustrate the echo form for structural visualization, and is not the actual input to the network. In the implementation, the network input consists of two channels corresponding to the real and imaginary parts of the radar echo signal, so as to preserve the complete complex-valued echo information.}
	\label{resnet18}
	\vspace{-0.4cm}
\end{figure}

\begin{figure*}[t]
	\centering
	\includegraphics[width=1.0\textwidth, height=0.145\textwidth]{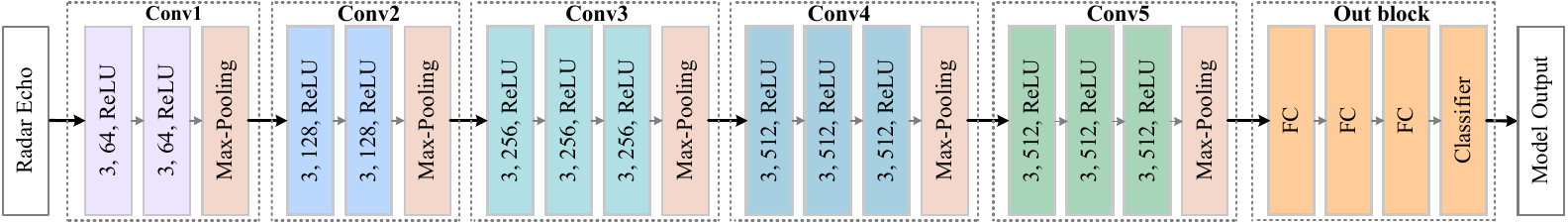}
	\vspace{-0.6cm}	
	\caption{Architecture of the VGG-based detection model. The model uses the same input and output settings as the ResNet-based detection model. The max-pooling layers downsample the resolution of the extracted feature maps, and FC denotes a fully connected layer.}
	\label{fig:vgg}
	\vspace{-0.4cm}
\end{figure*}

\subsection{Implementation Details} \label{section:implementation_details}
To enable the deep-learning-based detection model to process complex-valued radar echo signals, the real and imaginary parts of each echo sample are concatenated along the feature-channel dimension to form a two-channel real-valued input, which is then fed into the detection model. The model parameters are initialized using Kaiming initialization~\cite{10.1109/ICCV.2015.123} and optimized using stochastic gradient descent (SGD)~\cite{Loshchilov2017DecoupledWD}, with the learning rate set to 0.001. All experiments are conducted on a high-performance computing platform equipped with dual Intel(R) Xeon(R) Platinum 8380 CPUs, 1024 GB memory, and eight NVIDIA A40 GPUs. Unless otherwise specified, the number of pulses in each input sample is set to $P=64$, and all experimental results are reported as the average over five independent runs.

For the sample-dependent uniform error type (Type~\ref{type 2}), the label-error probability $\varphi_{\mathbf{x}_i}$ is determined by a sample-dependent function $g(\mathbf{x}_i)\in[0,1]$. To investigate the influence of different forms of $g$ on the robustness of SL-RLF under partial label errors, we construct the following four types of $g$.

\textbf{Uniform distribution:} To make the output of $g$ approximately follow $\text{Unif}[0,1]$, we construct $g$ using a three-layer multilayer perceptron (MLP), followed by a triangular-wave function $f_{\text{sw}}$ that maps the MLP output to $[0,1]$:
\begin{equation}
	f_{\text{sw}}(x) = 1 - \left| (x \bmod 2) - 1 \right|.
\end{equation}
Here, $\bmod$ denotes the modulo operation, and $\left| \cdot \right|$ denotes the absolute value. During MLP training, maximum mean discrepancy (MMD) is used to quantify the discrepancy between the output distribution of $g(\mathbf{x})$ and $\text{Unif}[0,1]$,  and this discrepancy is used as the loss function to guide optimization of the MLP parameters~\cite{li2017mmd}.

\textbf{Logarithmic distribution:} After training an MLP such that its output approximately follows $\text{Unif}[0,1]$, we further apply a monotonic transformation to construct a logarithmic-uniform-type distribution:
\begin{equation}
	g(\mathbf{x})=2^{\text{MLP}(\mathbf{x})} - 1.
\end{equation}
The derivation showing that this transformation induces the desired logarithmic-uniform-type distribution is provided in Appendix~4.

\textbf{Beta distribution:} According to the random-variable transformation formula, if $U \sim \text{Unif}[0,1]$, then $U^2$ has density
$f(y) = \frac{1}{2\sqrt{y}}$ on $[0,1]$, which follows the $\mathrm{Beta}\left(\frac{1}{2}, 1\right)$ distribution. Accordingly, we set
\begin{equation}
	g(\mathbf{x})=\text{MLP}(\mathbf{x}) \cdot \text{MLP}(\mathbf{x}).
\end{equation}
In this way, the output distribution of $g(\mathbf{x})$ approximately follows a $\mathrm{Beta}\left(0.5,1\right)$ distribution.

\textbf{Truncated Gaussian distribution:} To construct a Gaussian type error probability defined on $[0,1]$, we first specify the original Gaussian distribution as $\mathcal{N}\left(0.5,0.25^2 \right)$, such that its mean is located at the center of the interval. The distribution is then truncated to $[0,1]$ and normalized, resulting in the target probability density function
\begin{equation}
	p_{\mathrm{tg}}(z)
	=
	\frac{
		\phi\!\left(\dfrac{z-\mu}{\sigma}\right)
	}{
		\sigma\left[
		\Phi\!\left(\dfrac{1-\mu}{\sigma}\right)
		-
		\Phi\!\left(\dfrac{0-\mu}{\sigma}\right)
		\right]
	},
	\ z\in[0,1],
\end{equation}
where $\phi(\cdot)$ and $\Phi(\cdot)$ denote the probability density function and cumulative distribution function of the standard Gaussian distribution, respectively, and $\mu$ and $\sigma$ denote the mean and standard deviation of the original Gaussian distribution. Next, a three-layer MLP is employed to construct the mapping $g:\mathcal{X}\to[0,1]$. The MMD between the output distribution of $g(\mathbf{x})$ and the target distribution $p_{\mathrm{TG}}(z)$ is then minimized to optimize the MLP parameters, such that the output of $g(\mathbf{x})$ approximately follows the target distribution.

After obtaining the above mappings, a scaling factor $\alpha$ is introduced to control the upper bound of the label-error probability:
\begin{equation}
	\varphi_{\mathbf{x}_i} = \alpha g(\mathbf{x}_i) \le \alpha .
	\label{eq:upper_error}
\end{equation}
This setting enables experiments with different proportions of mislabeled samples in the training data and allows us to evaluate the robustness of each method with respect to the label-error probability.

\begin{table}[t]
	\renewcommand{\arraystretch}{1.3}
	\caption{HH-Polarized Data of MTDSP Dataset Under Sea State 2}
	\label{MTDSP HH 2}
	\centering
	\setlength{\tabcolsep}{3.5mm}{ 
		\vspace{-0.2cm}
		\begin{tabular}{cc}
			\bottomrule
			\cmidrule{1-2} 
			20221115040022\_stare\_HH.mat & 20221115050027\_stare\_HH.mat \\ 
			20221115070036\_stare\_HH.mat & 20221115080041\_stare\_HH.mat \\
			\bottomrule
			\cmidrule{1-2}
	\end{tabular}}
	\vspace{-0.6cm}
\end{table}
\subsection{Evaluation Metrics}
To evaluate the robustness of different methods against partial label errors, we train the detection model using each method on the same partially mislabeled training datasets and evaluate its detection performance on accurately labeled test datasets. Under the same model architecture and training settings, higher test performance indicates stronger robustness to partial label errors during training.

For the MTDSP dataset, this detection task contains only two classes, namely clutter and target. Therefore, the detection probability at a fixed false-alarm rate is adopted as the evaluation metric and denoted by $P_{\mathrm{d}}^{\mathrm{fa}}$. This metric reflects the target detection capability under a specified false-alarm constraint and is therefore appropriate for evaluating binary radar detection performance. Unless otherwise specified, the false-alarm rate is set to $P_{\mathrm{fa}}=0.0001$, and the corresponding metric is denoted by $P_{\mathrm{d}}^{0.0001}$. In the implementation, we follow the procedure in~\cite{11124398} to compute $P_{\mathrm{d}}^{0.0001}$.

For the LSS-FMCWR dataset, which contains multiple target classes, maintaining a constant false-alarm rate across classes is less straightforward than in binary detection tasks. Therefore, the macro-averaged F1 score (Macro-F1) is adopted as the evaluation metric. This metric jointly accounts for the precision and recall of each class and thus provides a comprehensive characterization of multiclass detection performance. It is defined as
\begin{equation}
	\hspace{-0.1cm}
	\begin{aligned}
		&\text{Precision}_k\!=\!\frac{\text{TP}_k}{\text{TP}_k+\text{FP}_k},\qquad \text{Recall}_k\!=\!\frac{\text{TP}_k}{\text{TP}_k+\text{FN}_k}, \\
		&\text{F1}_k\!=\!\frac{2\cdot \text{Precision}_k \cdot \text{Recall}_k}{\text{Precision}_k+\text{Recall}_k},
		\ {\text{Macro-F1}}\!=\!\frac{1}{K}\sum_{k=1}^{K}\text{F1}_k.
	\end{aligned}
\end{equation}
Here, $\text{TP}_k$ denotes the number of samples correctly classified as class $k$, $\text{FP}_k$ denotes the number of samples from other classes incorrectly classified as class $k$, and $\text{FN}_k$ denotes the number of class $k$ samples incorrectly classified as other classes.

\subsection{Performance Comparison}
To evaluate the robustness of the proposed SL-RLF under partial label errors in the training data, we select CE~\cite{lecun2015deep}, MSE~\cite{dohare2024loss}, Co-teaching~\cite{han2018co}, and TCE~\cite{10.5555/3491440.3491745} as comparison methods. Specifically, CE and MSE are commonly used loss functions in standard supervised learning frameworks and are typically adopted for model training with accurate labels; thus, they serve as non-robust baselines. Co-teaching is a robust learning framework based on training-strategy modification, which mitigates the adverse effects of mislabeled samples on model parameter updates through dual-network collaborative training and a sample-selection strategy. TCE is a robust learning framework based on loss-function design and follows a technical route similar to that of the proposed SL-RLF, namely improving the robustness of detection models to partial label errors by modifying the loss function in the standard supervised learning framework.

The experiments are conducted by constructing accurately labeled training, validation, and test sets based on the data corresponding to Table~\ref{MTDSP HH 2}. The labels in the training set are then perturbed according to the probabilistic transition models defined in Type~\ref{type 1}, Type~\ref{type 2}, and Type~\ref{type 3}, respectively, to generate different types of label errors.

\begin{figure}[t]
	\centering
	\includegraphics[width=0.48\textwidth, height=0.28\textwidth]{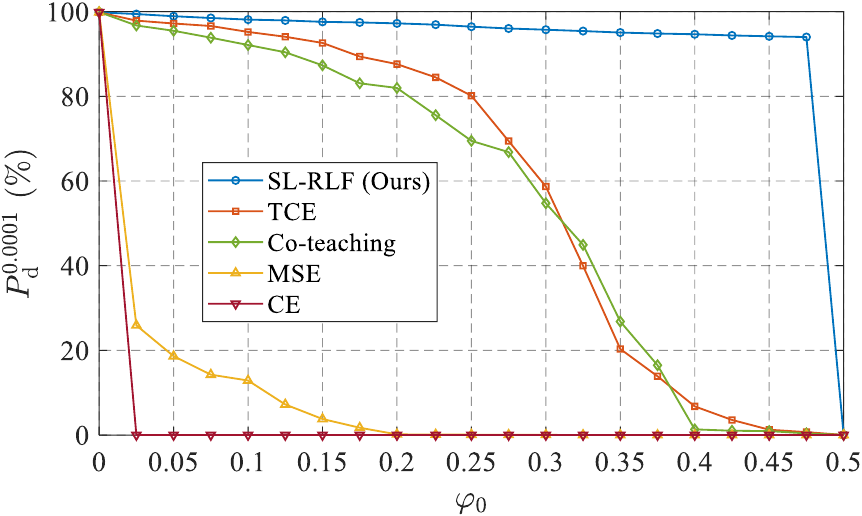}
	\vspace{-0.2cm}
	\caption{Detection performance on accurately labeled test data for models trained by different methods with partially mislabeled training data under constant uniform label errors (Type~\ref{type 1}) in a two-class detection task.}
	\label{fig:Uniform_data_comper}
	\vspace{-0.5cm}
\end{figure}

{\bf{Constant uniform error (Type~\ref{type 1}):}} Fig.~\ref{fig:Uniform_data_comper} presents the performance comparison of different methods under constant uniform label errors. Because CE and MSE do not include mechanisms for mitigating label errors, their performance decreases markedly as the error probability $\varphi_0$ increases. This result indicates that conventional loss functions designed for accurately labeled data have limited robustness in this setting. In contrast, Co-teaching and TCE mitigate the performance degradation caused by label errors to some extent; however, their performance still decreases noticeably as $\varphi_0$ further increases. Compared with these methods, the proposed SL-RLF maintains higher and more stable detection performance over the range $\varphi_0<0.5$, demonstrating stronger robustness to partial label errors.

Since this experiment considers a binary detection task, the number of classes is $K=2$. According to Theorem~\ref{theorem 1}, under the constant uniform error mechanism, the condition under which a robust learning framework constructed with a symmetric loss function guarantees theoretical robustness to partial label errors is $\varphi_0 < \left( K-1 \right) / K = 0.5$. Because the proposed SL-RLF employs $\ell_{\mathrm{ROB}}$, which satisfies the symmetry condition, as the training loss, this condition directly specifies the theoretical robustness range of SL-RLF under this error mechanism. As shown in Fig.~\ref{fig:Uniform_data_comper}, the range over which SL-RLF remains robust in the experiments is generally consistent with this theoretical bound, thereby validating the theoretical analysis in Theorem~\ref{theorem 1}.

\begin{figure}[t]
	\centering
	\includegraphics[width=0.48\textwidth, height=0.29\textwidth]{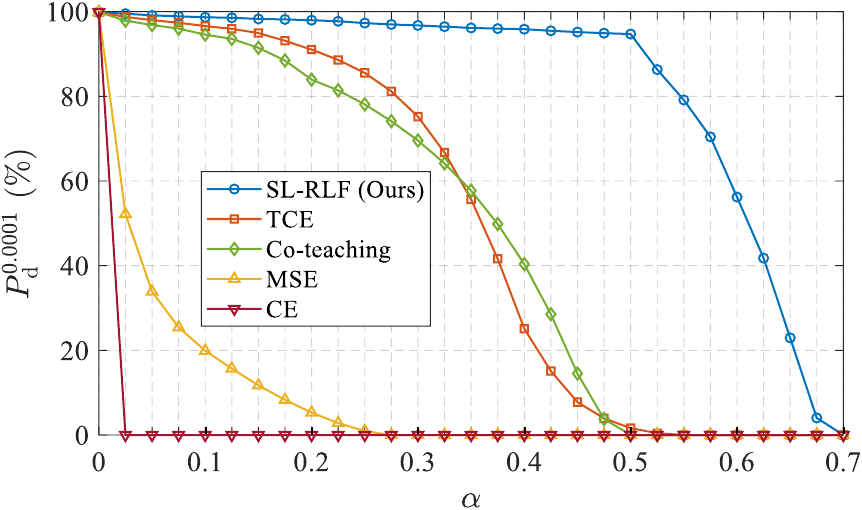}
	\vspace{-0.2cm}
	\caption{Detection performance on accurately labeled test data for models trained by different methods with partially mislabeled training data under sample-dependent uniform error (Type~\ref{type 2}) in a two-class detection task.}
	\label{fig:NonUniform_data_comper}
	\vspace{-0.1cm}
\end{figure}

\begin{figure}[t]
	\centering
	\includegraphics[width=0.49\textwidth, height=0.53\textwidth]{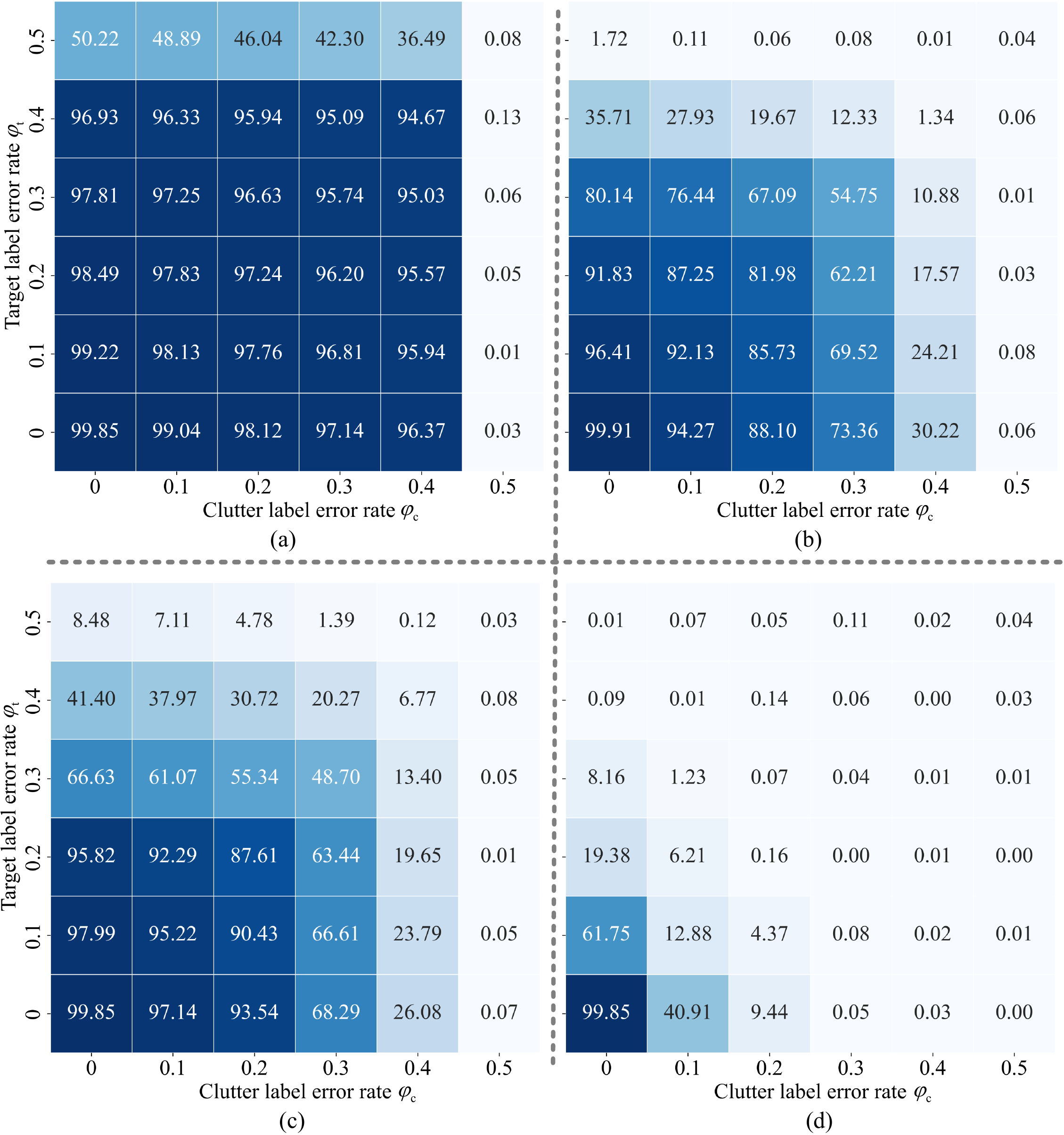}
	\vspace{-0.7cm}
	\caption{Detection performance on accurately labeled test data for models trained by different methods with partially mislabeled training data under class-dependent error (Type~\ref{type 3}) in a two-class detection task. (a) The proposed SL-RLF; (b) Co-teaching; (c) TCE; (d) MSE.}
	\label{fig:class_data_comper}
	\vspace{-0.5cm}
\end{figure}

It should be noted that Theorem~\ref{theorem 1} characterizes robustness to partial label errors in terms of population risk. In practical training, however, the detection model can only compute  the empirical risk over finite training samples as an approximation to the population risk, and the resulting solution usually only approximates the population-risk minimizer~\cite{bottou2018optimization}. Therefore, even when $\varphi_0$ lies within the theoretically valid range, the learned model may deviate from the population-risk minimizer in finite-sample training. Moreover, as the proportion of mislabeled samples increases, the amount of reliable supervisory information available for effective parameter updates decreases, which may lead to slight degradation in detection performance. Accordingly, the performance attenuation observed in Fig.~\ref{fig:Uniform_data_comper} over the range $\varphi_0<0.5$ does not contradict the theoretical conclusion.

{\bf{Sample-dependent uniform error (Type~\ref{type 2}):}} When constructing the sample-dependent uniform label errors, we first use a function $g$ whose output follows a uniform distribution to generate the label-error probability of each sample, and then control the upper bound of the error probability through the scaling factor $\alpha$ according \eqref{eq:upper_error}. The corresponding experimental results are shown in Fig.~\ref{fig:NonUniform_data_comper}. The results show that, when the condition required by Theorem~\ref{theorem 2}, i.e.,
\begin{equation}
	\max_{\mathbf{x}_i} \varphi_{\mathbf{x}_i} = \alpha < \frac{K-1}{K} = 0.5
\end{equation}
is satisfied, SL-RLF exhibits strong robustness and achieves overall better performance than comparison methods such as Co-teaching and TCE. This result is consistent with the theoretical analysis in Theorem~\ref{theorem 2}. 

Notably,  when $\alpha\ge0.5$, Fig.~\ref{fig:NonUniform_data_comper} shows that SL-RLF still retains a certain degree of robustness. This observation does not contradict Theorem~\ref{theorem 2}, because the theorem provides a sufficient condition for robustness rather than a necessary one. Therefore, when the condition is not satisfied, the theoretical guarantee no longer holds, but this does not imply that the robustness of the method must immediately fail.

In addition, under sample-dependent uniform errors, SL-RLF does not exhibit an abrupt performance change near the theoretical threshold $\alpha=0.5$, as observed in the constant uniform error case  (Type~\ref{type 1}). This is because, under the current experimental setting, $\alpha$ only controls the upper bound of the sample error probability $\varphi_{\mathbf{x}_i}$, whereas the error probability of each sample is generated by $g(\mathbf{x}_i)$. When $\alpha\geq 0.5$, some samples still satisfy $\varphi_{\mathbf{x}_i}<0.5$, and these samples can continue to provide effective supervisory information. In contrast, under constant uniform errors, once $\varphi_0=0.5$, the labels of all samples are flipped with probability $0.5$. In this case, the training labels no longer contain effective discriminative information about the true labels, making it difficult for the model to maintain robustness. As $\alpha$ further increases in the sample-dependent uniform errors type, the proportion of samples satisfying $\varphi_{\mathbf{x}_i}>0.5$ gradually grows. These samples are more likely to provide misleading supervisory signals, thereby causing a significant degradation in model robustness. This phenomenon is consistent with the sufficient-condition nature of Theorem~\ref{theorem 2} and the theoretical analysis in Appendix~2.

\begin{figure}[t]
	\centering
	\includegraphics[width=0.48\textwidth, height=0.29\textwidth]{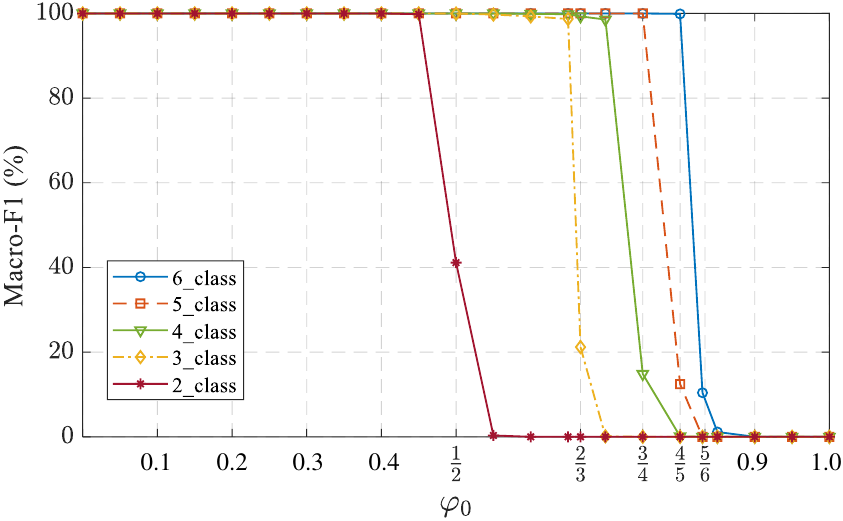}
	\vspace{-0.2cm}
	\caption{Effect of the number of target classes on the detection performance of the proposed method for data with constant uniform errors (Type~\ref{type 1}).}
	\label{fig:class_ablation}
	\vspace{-0.2cm}
\end{figure}

{\bf{Class-dependent error (Type~\ref{type 3}):}} Fig.~\ref{fig:class_data_comper} presents the robustness comparison of different methods under class-dependent label errors. Compared with methods such as Co-teaching and TCE, SL-RLF still maintains superior robustness in this setting. Since this experiment is a two-class detection task, the effective robustness range described in Theorem~\ref{theorem 3} can be simplified to $\varphi_{y_{\mathbf{x}_i}}<0.5$. As shown in Fig.~\ref{fig:class_data_comper}, the experimentally observed performance trend is generally consistent with the theoretical range given by Theorem~\ref{theorem 3}. Specifically, when either the label-error probability of the clutter class $\varphi_{\mathrm{c}}$ or that of the target class $\varphi_{\mathrm{t}}$ approaches the theoretical boundary of $0.5$, the performance of the detection model exhibits pronounced degradation, indicating that the reliability of the supervisory information is substantially weakened near this boundary.

\begin{figure}[t]
	\centering
	\includegraphics[width=0.48\textwidth, height=0.29\textwidth]{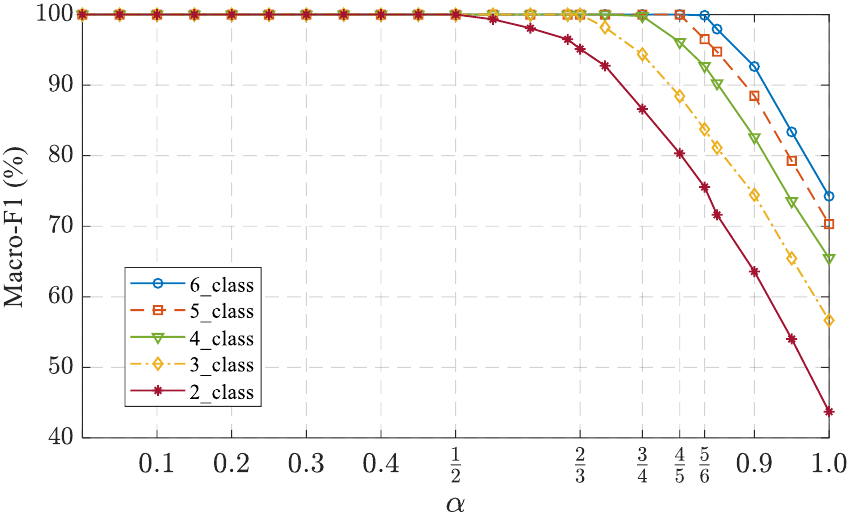}
	\vspace{-0.2cm}
	\caption{Effect of the number of target classes on the detection performance of the proposed method for data with sample-dependent uniform errors (Type~\ref{type 2}).}
	\label{fig:class_ablation_non_uniform}
	\vspace{-0.4cm}
\end{figure}

In addition, Fig.~\ref{fig:class_data_comper}(a) shows that when the label-error rate of the target class reaches $\varphi_{\mathrm{t}}=0.5$, the performance degradation of SL-RLF is less severe than that observed when the label-error rate of the clutter class reaches $\varphi_{\mathrm{c}}=0.5$. This phenomenon is mainly related to class imbalance in the training data. Since the number of clutter samples is significantly larger than that of target samples, when the labels of target samples are flipped with probability $0.5$, the samples affected by this high error probability are mainly from the relatively smaller target class. Consequently, the training set still contains a large number of correctly labeled clutter samples, which provide reliable supervisory information for model training. In contrast, when the label-error rate of the clutter class reaches $\varphi_{\mathrm{c}}=0.5$, the dominant clutter samples are affected by a high error rate, which more substantially degrades the overall quality of supervisory information and leads to a more pronounced performance decline.

\subsection{Ablation Experiments}
{\bf{Number of target classes:}} To analyze how the total number of classes $K$ distinguished by the detection model affects the robustness of SL-RLF, we construct training datasets with different class-number settings based on the LSS-FMCWR dataset and perturb the training labels according to the constant uniform error model defined in Type~\ref{type 1}. In the experiments, Macro-F1 is used to evaluate the detection performance of the model on multiclass target datasets. The results in Fig.~\ref{fig:class_ablation} show that SL-RLF maintains robustness under different class-number settings, and the admissible range of the error probability $\varphi_0$ over which robustness is maintained expands as the number of classes increases. This trend is consistent with the condition $\varphi_0 < {\left( K-1\right) }/{K}$ given in Theorem~\ref{theorem 1}, further validating the applicability of Theorem~\ref{theorem 1} to multiclass target detection.

To further investigate the effect of the number of classes on the robustness range under sample-dependent uniform errors, we construct partially mislabeled datasets according to Type~\ref{type 2}. The sample-dependent function $g(\mathbf{x})$ is set such that its output follows a uniform distribution. The experimental results are shown in Fig.~\ref{fig:class_ablation_non_uniform}. When $\alpha$ lies within the theoretical range specified by Theorem~\ref{theorem 2}, SL-RLF maintains stable detection performance. In contrast, when $\alpha$ exceeds the corresponding theoretical range, the detection performance decreases to varying degrees. These results indicate that the effect of the number of classes on the robustness interval of SL-RLF is consistent with the analysis in Theorem~\ref{theorem 2}.

Overall, these results confirm that the robustness of the SL-RLF constructed with $\ell_{\mathrm{ROB}}$ generalizes across different class-number settings and uniform label-error mechanisms. They further show that the effective robustness ranges of SL-RLF are well characterized by the theoretical conditions established in the corresponding theorems.

\begin{figure}[t]
	\centering
	\includegraphics[width=0.48\textwidth, height=0.30\textwidth]{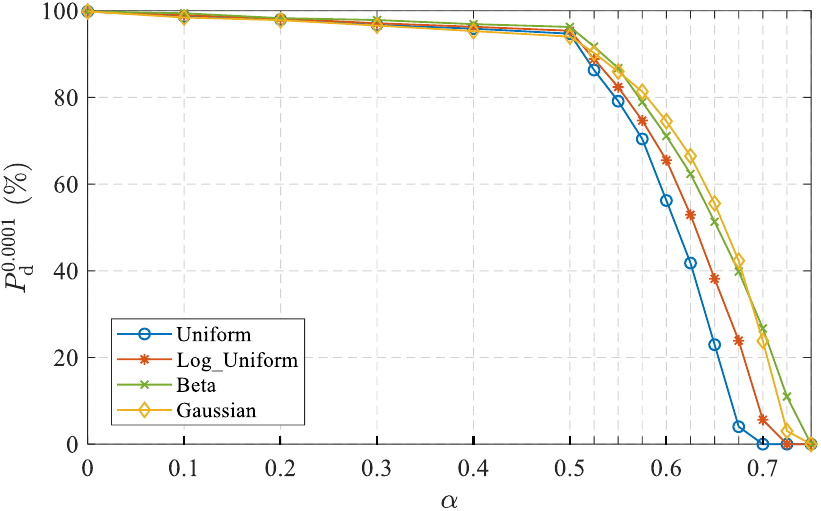}
	\vspace{-0.2cm}
	\caption{Detection performance of the proposed method under different distributions of the error-probability generation function $g$ for data with sample-dependent uniform errors (Type~\ref{type 2}).}
	\label{fig:g_ablation}
	\vspace{-0.4cm}
\end{figure}

{\bf{Type of $g$:}} According to the definition of Type~\ref{type 2}, under the sample-dependent uniform error type, the error probability $\varphi_{\mathbf{x}_i}$ of each sample is determined by the sample-dependent function $g(\mathbf{x}_i)\in[0,1]$. Theorem~\ref{theorem 2} indicates that, when the corresponding assumptions hold, the robustness guarantee for SL-RLF is independent of the specific form of $g$. To verify this theoretical result, we conduct ablation experiments using four sample-dependent functions $g$ whose outputs follow different distributions, including uniform, logarithmic, Beta, and truncated Gaussian distributions, as defined in Section~\ref{section:implementation_details}.

The results in Fig.~\ref{fig:g_ablation} show that, within the range $\alpha<0.5$ required by Theorem~\ref{theorem 2}, SL-RLF maintains stable robustness for different types of $g$. This result is consistent with the theoretical analysis, confirming that the robustness guarantee for SL-RLF does not depend on the specific form of $g$. Furthermore, when $\alpha>0.5$, the degree of performance degradation differs across different choices of $g$. This is because there exist samples that satisfy $\varphi_{\mathbf{x}_i}>0.5$ and thus more frequently provide misleading supervisory signals. Since different choices of $g$ induce different distributions of the error probability, the proportion of samples satisfying $\varphi_{\mathbf{x}_i}>0.5$ also varies, leading to different degrees of performance degradation.

\begin{table}[t]
	\renewcommand{\arraystretch}{1.3}
	\caption{MTDSP data under different polarization modes and sea state.}
	\label{MTDSP pol level}
	\centering
	\setlength{\tabcolsep}{1.2mm}{ 
		\vspace{-0.2cm}
		\begin{tabular}{cccc}
			\bottomrule
			\cmidrule{1-4} 
			{\bf{Mode}} & {\bf{Level}} & \multicolumn{2}{c}{{\bf{Data Name}}} \\
			\bottomrule
			\cmidrule{1-4}
			\multirow{8}{*}{HH} & \multirow{2}{*}{5} & 20221112180016\_stare\_HH & 20221112205109\_stare\_HH \\ 
			& & 20221112220105\_stare\_HH & 20221113030022\_stare\_HH\\
			\cmidrule{2-4} 
			& \multirow{2}{*}{4} & 20221112140039\_stare\_HH & 20221113090050\_stare\_HH \\ 
			& & 20221113170002\_stare\_HH & 20221114000105\_stare\_HH\\
			\cmidrule{2-4} 
			& \multirow{2}{*}{3} & 20221114180042\_stare\_HH & 20221114060027\_stare\_HH \\ 
			& & 20221114150028\_stare\_HH & 20221114050022\_stare\_HH\\
			\cmidrule{2-4} 
			& 2 & \multicolumn{2}{c}{See Table~\ref{MTDSP HH 2}}\\
			\bottomrule
			\cmidrule{1-4}
			\multirow{9}{*}{VV} & \multirow{2}{*}{5} & 20221112190032\_stare\_VV & 20221112210042\_stare\_VV \\ 
			& & 20221112220047\_stare\_VV & 20221113050014\_stare\_VV\\
			\cmidrule{2-4} 
			& \multirow{2}{*}{4} & 20221112171210\_stare\_VV & 20221113140029\_stare\_VV \\ 
			& & 20221113170004\_stare\_VV & 20221113230032\_stare\_VV\\
			\cmidrule{2-4} 
			& \multirow{2}{*}{3} & 20221114010041\_stare\_VV & 20221114040055\_stare\_VV \\ 
			& & 20221114060105\_stare\_VV & 20221114140023\_stare\_VV\\
			\cmidrule{2-4} 
			& \multirow{2}{*}{2} & 20221115040032\_stare\_VV & 20221115050036\_stare\_VV \\
			& & 20221115070046\_stare\_VV & 20221115080050\_stare\_VV\\
			\bottomrule
			\cmidrule{1-4}
	\end{tabular}}
	\vspace{-0.4cm}
\end{table}

\begin{figure}[t]
	\centering
	\includegraphics[width=0.48\textwidth, height=0.33\textwidth]{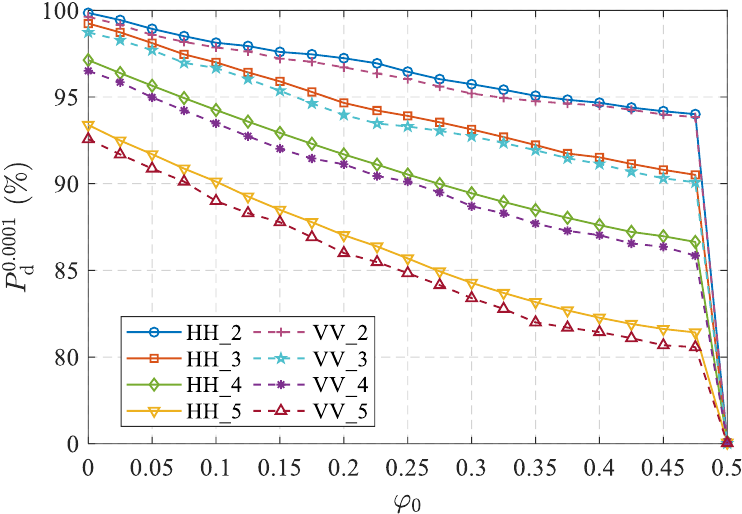}
	\vspace{-0.2cm}
	\caption{Detection performance of the proposed method under different sea-state levels and polarization modes.}
	\label{pol_level_ablation}
	\vspace{-0.4cm}
\end{figure}

{\bf{Sea state and polarization mode:}} To analyze the influence of sea state and radar polarization mode on the robustness of SL-RLF, we use the data in Table~\ref{MTDSP pol level} and construct datasets with constant uniform errors under different sea states and polarization modes according to the probabilistic transition model defined in Type~\ref{type 1}. The results in Fig.~\ref{pol_level_ablation} show that SL-RLF maintains favorable robustness on data corresponding to different polarization modes and sea-state conditions. Moreover, the ranges of error probability over which robustness is observed are consistent with the theoretical range specified in Theorem~\ref{theorem 1}.

Furthermore, as the sea state increases, the robustness of SL-RLF shows a decreasing trend. This phenomenon is mainly because the sea-clutter background becomes more complex under higher sea states, and target echoes are more likely to be submerged in strong sea clutter. As a result, the separability between targets and background is reduced, and the detection model faces greater difficulty in learning effective decision boundaries. In this case, the adverse effect of partial label errors on model parameter learning is further amplified, leading to more pronounced performance degradation.

\begin{figure}[t]
	\centering
	\includegraphics[width=0.48\textwidth, height=0.33\textwidth]{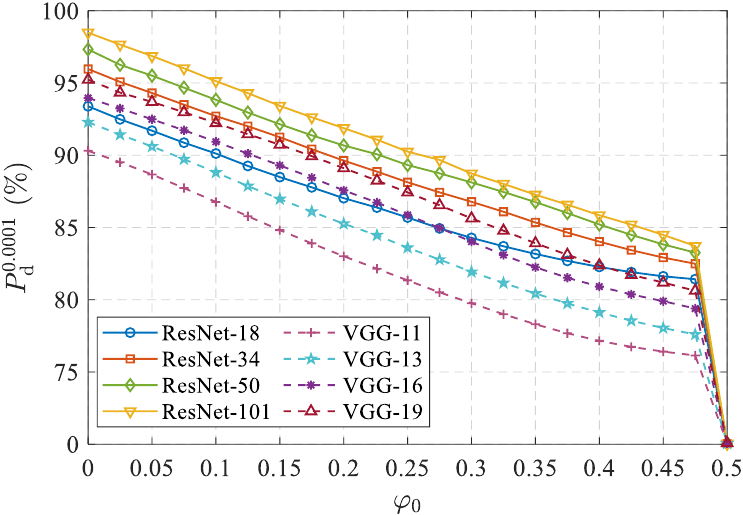}
	\vspace{-0.2cm}
	\caption{Detection performance of the proposed method under different model architectures. Here, “ResNet-18” indicates that the detection model uses the ResNet network architecture, with 18 denoting the number of network layers. A larger number corresponds to a deeper network architecture, typically involving more parameters and stronger detection capability.}
	\label{fig:model_ablation}
	\vspace{-0.4cm}
\end{figure}

{\bf{Model architecture:}} To validate the applicability of the proposed SL-RLF framework to different network architectures of the detection model, we construct two detection models separately based on the VGG- and ResNet-based architectures described in Section~\ref{section:Model_Architecture}. Based on the radar echo data under sea state 5 with HH polarization in Table~\ref{MTDSP pol level}, partially mislabeled training data are generated under the constant uniform error setting (Type~\ref{type 1}) for the ablation experiment. The corresponding results are shown in Fig.~\ref{fig:model_ablation}. The results indicate that, across different network architectures and model depths, training with the proposed SL-RLF framework improves the robustness of detection models against partial label errors.

\section{Conclusion} \label{section:Conclusion}
This paper investigated the training problem of deep-learning-based radar target detection models in the presence of partially mislabeled training data and proposed a robust learning framework to mitigate the detection performance degradation caused by label errors. To characterize label-error generation mechanisms, we first established corresponding probabilistic models. Based on these models, we analyzed, within a Bayes risk minimization formulation, the conditions under which the loss function in the learning framework enables robustness against label errors. Theoretical results showed that the symmetry condition of the loss function provides an important sufficient condition for robustness. Motivated by this result, we designed a symmetric robust loss function $\ell_{\text{ROB}}$ and constructed a robust learning framework, termed SL-RLF, for radar target detection with partially mislabeled training data. Under the corresponding label-error conditions, the optimal detection model learned by SL-RLF from partially mislabeled training data can achieve the same detection performance under the true data distribution as the optimal model learned from accurately labeled data.

Experimental results on measured radar echo datasets demonstrated that SL-RLF achieves stronger robustness than representative competing methods, including CE, MSE, Co-teaching, and TCE. Ablation experiments further verified the applicability of the proposed SL-RLF framework across different detection model architectures, sea states, polarization modes, and numbers of target classes. Moreover, the experimentally observed robustness trends were consistent with the theoretical analysis, supporting the validity of the theoretical results.

\bibliographystyle{IEEEtran}
\bibliography{IEEEabrv,myref}

\clearpage
\appendices 
\setcounter{theorem}{0} 
{
	\section{Proof of Theorem 1}
	\begin{theorem}
		Let $\ell$ be a symmetric loss function. Under the Type~1 error model, i.e., constant uniform error, if
		\begin{equation}
			\varphi_0 < \frac{K-1}{K},
		\end{equation}
		then $\ell$ is robust to data with this type of incorrect labels.
	\end{theorem}
	
	\textbf{Proof:} For a dataset $\hat{\mathcal{D}}$ satisfying the constant uniform error model, the error probability $\varphi_{\mathbf{x}}$ is the same constant $\varphi_0$ for all samples. Therefore, by the law of total expectation, $\hat{R}_\ell(f_{\boldsymbol{\theta}})$ can be expanded as
	\begin{align}
		&\hat{R}_\ell(f_{\boldsymbol{\theta}}) = \mathbb{E}_{\mathbf{x}, \hat{y}_{\mathbf{x}}} \left[ \ell(f_{\boldsymbol{\theta}}\left( \mathbf{x}\right), \mathbf{e}_{\hat{y}_{\mathbf{x}}})\right]  \nonumber\\
		&= \mathbb{E}_{\mathbf{x}, y_{\mathbf{x}} }\left[  \mathbb{E}_{\hat{y}_{\mathbf{x}} | \mathbf{x}, y_{\mathbf{x}} }\left[ \ell\left( f_{\boldsymbol{\theta}}\left( \mathbf{x}\right), \mathbf{e}_{\hat{y}_{\mathbf{x}}}\right) \right] \right]  \label{eq:proof1_1}\\
		&=  \mathbb{E}_{\mathbf{x} , {y}_{\mathbf{x}}} \!\!\left[(1 \!-\! \varphi_0) \ell\left( f_{\boldsymbol{\theta}}\left( \mathbf{x}\right), \mathbf{e}_{y_{\mathbf{x}}}\right)  +\! \frac{\varphi_0}{K\!-\!1} \!\sum_{j \neq y_{\mathbf{x}}} \!\!\ell\!\left( f_{\boldsymbol{\theta}}\left( \mathbf{x}\right), \mathbf{e}_j\right) \right]\!.\nonumber
	\end{align}
	Since $\ell$ is a symmetric loss function, it follows from Definition~2 that
	\begin{equation}
		\sum_{k=0}^{K} \ell(f_{\boldsymbol{\theta}}\left( \mathbf{x}\right), \mathbf{e}_k) = C, \quad \forall \mathbf{x} \in \mathcal{X}.
		\label{eq:symmetry_loss}
	\end{equation}
	Hence, for any $(\mathbf{x}, y_{\mathbf{x}})$,
	\begin{equation}
		\sum_{j\neq y_{\mathbf{x}}}\ell\left(f_{\boldsymbol{\theta}}(\mathbf{x}), \mathbf{e}_j\right)
		= C-\ell\left(f_{\boldsymbol{\theta}}(\mathbf{x}), \mathbf{e}_{y_{\mathbf{x}}}\right).
	\end{equation}
	Substituting this relation into \eqref{eq:proof1_1} yields
	\begin{equation}
		\begin{aligned}
			\hat{R}_\ell(f_{\boldsymbol{\theta}}) &= (1 - \varphi_0) R_\ell(f_{\boldsymbol{\theta}}) + 		\frac{\varphi_0}{K-1} (C - R_\ell(f_{\boldsymbol{\theta}})) \\
			&= \frac{C \varphi_0}{K-1} + \left(1 - \frac{\varphi_0 K}{K-1}\right) 	R_\ell(f_{\boldsymbol{\theta}}).
		\end{aligned}
	\end{equation}
	Let $f_{\boldsymbol{\theta}^\ast}$ be a model attaining the global minimum of $R_\ell(f_{\boldsymbol{\theta}})$. Then,
	\begin{equation}
		R_\ell(f_{\boldsymbol{\theta}^\ast}) \leq R_\ell(f_{\boldsymbol{\theta}}), \quad \forall f_{\boldsymbol{\theta}}.
	\end{equation}
	Furthermore, from the condition $\varphi_0 < \left( {K-1}\right) /{K}$, we have
	\begin{equation}
		\left(1 - \frac{\varphi_0 K}{K-1}\right) > 0 .
	\end{equation}
	Therefore, for any $f_{\boldsymbol{\theta}}$,
	\begin{equation}
		\hspace{-0.1cm}
		\hat{R}_\ell(f_{\boldsymbol{\theta}^\ast}) \!-\! \hat{R}_\ell(f_{\boldsymbol{\theta}})\!=\!\left(1 \!-\! \frac{\varphi_0 K}{K\!-\!1}\right) \!\!\left( R_\ell(f_{\boldsymbol{\theta}^\ast}) \!-\! R_\ell(f_{\boldsymbol{\theta}}) \right)\!\leq\!0.
		\label{eq:leq0}
	\end{equation}
	This follows that, under the constant uniform error type, an optimal solution of ${R}_\ell(f_{\boldsymbol{\theta}})$ is also an optimal solution of $\hat{R}_\ell(f_{\boldsymbol{\theta}})$. If
	\begin{equation}
		\hat{f}_{\boldsymbol{\theta}^\ast}= \arg\min_{{\boldsymbol{\theta}}}\hat{R}_\ell(f_{\boldsymbol{\theta}}),
	\end{equation}
	Accordingly, the optimal model under incorrect labels can be represented as $\hat{f}_{\boldsymbol{\theta}^\ast}=f_{\boldsymbol{\theta}^\ast}$, and consequently
	\begin{equation}
		\hat{R}_\ell(f_{\boldsymbol{\theta}^\ast}) = \hat{R}_\ell(\hat{f}_{\boldsymbol{\theta}^\ast}).
	\end{equation} 
	Thus, the robustness condition in Definition~1 is satisfied. Hence, when $\varphi_0 < \left( {K-1}\right) /{K}$, the symmetric loss function is robust to this type of incorrect labels.
	
	It is worth noting that the boundary case $\varphi_0 = \left( {K-1}\right) /{K}$ requires separate discussion. Although in this case $\hat{R}_\ell(f_{\boldsymbol{\theta}^\ast})-\hat{R}_\ell(f_{\boldsymbol{\theta}}) = 0$, so that \eqref{eq:leq0} still holds, this does not imply that the symmetric loss function remains resistant to incorrect labels at the boundary. In fact, the total risk reduces to
	\begin{equation}
		\hat{R}_\ell(f_{\boldsymbol{\theta}}) = \frac{C \varphi_0}{K-1} = \frac{C}{K}
	\end{equation} 
	Thus, $\hat{R}_\ell(f_{\boldsymbol{\theta}})$ degenerates into a constant independent of the model $f_{\boldsymbol{\theta}}$. In other words, any model can attain the same minimum total risk, including both models with correct discriminative capability and models without effective discriminative capability. Therefore, when $\varphi_0 = \left( {K-1}\right) /{K}$, minimizing $\hat{R}_\ell$ can no longer guarantee recovery of the optimal decision boundary obtained under accurate labels. Consequently, the loss function cannot be regarded as robust to incorrect labels under this boundary condition.
	
	This completes the proof. $\hfill\blacksquare$

	\section{Proof of Theorem 2}
	\begin{theorem}
		Let $\ell$ be a symmetric loss function. Under the Type~2 error model, namely, sample-dependent uniform error, if
		\begin{equation}
			\varphi_{\mathbf{x}} < \frac{K-1}{K}, \quad \forall \mathbf{x}\in\mathcal{X},
		\end{equation}
		and $R_\ell(f_{\boldsymbol{\theta}^\ast}) = 0$, then $\ell$ is robust to incorrect labels generated by this error model.
	\end{theorem}
	
	\textbf{Proof:} For any $ f_{\boldsymbol{\theta}} $, the population risk associated with the dataset $\hat{\mathcal{D}}$ under the sample-dependent uniform error model can be expanded by the law of total expectation as
	\begin{align}
		&\hat{R}_\ell(f_{\boldsymbol{\theta}})  = \mathbb{E}_{\mathbf{x}, \hat{y}_{\mathbf{x}}} \left[ \ell(f_{\boldsymbol{\theta}}\left( \mathbf{x}\right), \mathbf{e}_{\hat{y}_{\mathbf{x}}})\right] \nonumber \\
		&=  \mathbb{E}_{\mathbf{x},y_{\mathbf{x}}} \left[ \mathbb{E}_{\hat{y}_{\mathbf{x}} | \mathbf{x}, y_{\mathbf{x}}}\left[  \ell(f_{\boldsymbol{\theta}}\left( \mathbf{x}\right), \mathbf{e}_{\hat{y}_{\mathbf{x}}}) \right]  \right] \label{eq:expectation_fully_unfolded}\\
		&= \mathbb{E}_{\mathbf{x}, y_{\mathbf{x}}} \!\left[(1 - \varphi_{\mathbf{x}}) \ell(f_{\boldsymbol{\theta}}\left( \mathbf{x}\right), \mathbf{e}_{y_{\mathbf{x}}}) \!+\!\! \sum_{j \neq y_{\mathbf{x}}}\! \frac{\varphi_{\mathbf{x}} \ell(f_{\boldsymbol{\theta}}\left( \mathbf{x}\right), \mathbf{e}_j)}{K-1}\right].\nonumber
	\end{align} 
	Using \eqref{eq:symmetry_loss}, which follows from the symmetry condition in Definition~2, we obtain
	\begin{equation} 
		\sum_{j \neq y_{\mathbf{x}}} \frac{\varphi_{\mathbf{x}} \ell(f_{\boldsymbol{\theta}}\left( 	\mathbf{x}\right), \mathbf{e}_j)}{K-1} =  \frac{ \varphi_{\mathbf{x}} \left( C - \ell\left( f_{\boldsymbol{\theta}}\left( \mathbf{x}\right), \mathbf{e}_{y_{\mathbf{x}}}\right) \right)}{K-1}.
	\end{equation}
	Substituting this relation into \eqref{eq:expectation_fully_unfolded} yields
	\begin{align}
		&\hat{R}_\ell(f_{\boldsymbol{\theta}}) = \mathbb{E}_{\mathbf{x}, y_{\mathbf{x}}} \left[(1 - 		\varphi_{\mathbf{x}}) \ell(f_{\boldsymbol{\theta}}\left( \mathbf{x}\right), \mathbf{e}_{y_{\mathbf{x}}}) \right] \nonumber \\ 
		&\qquad\quad+ \mathbb{E}_{\mathbf{x}, y_{\mathbf{x}}} \left[\frac{ \varphi_{\mathbf{x}}\left( C - 	\ell\left( 	f_{\boldsymbol{\theta}}\left( \mathbf{x}\right), \mathbf{e}_{y_{\mathbf{x}}}\right) \right)}{K-1}  \right]\\
		&= \mathbb{E}_{\mathbf{x}, y_{\mathbf{x}}}\left[  \frac{C\varphi_{\mathbf{x}}}{K-1} \right]  + 	\mathbb{E}_{\mathbf{x}, y_{\mathbf{x}}} \left[ \left(1 - \frac{K \varphi_{\mathbf{x}}}{K-1}\right) \ell(f_{\boldsymbol{\theta}}\left( \mathbf{x}\right), \mathbf{e}_{y_{\mathbf{x}}})\right].\nonumber
	\end{align}
	By the condition $R_\ell(f_{\boldsymbol{\theta}^\ast}) = 0$ and the nonnegativity of the loss function, we have
	\begin{align}
		\ell(f_{\boldsymbol{\theta}}\left( \mathbf{x}\right), \mathbf{e}_{y_{\mathbf{x}}}) \geq 0, \quad 	\forall \mathbf{x} \in \mathcal{X}, \label{eq:geq_0}\\
		\ell(f_{\boldsymbol{\theta}^\ast}(\mathbf{x}), \mathbf{e}_{y_{\mathbf{x}}}) = 0, \quad \forall 	\mathbf{x} \in \mathcal{X}.
	\end{align}
	Therefore, for any $ f_{\boldsymbol{\theta}} $, we have
	\begin{align}
		&\hat{R}_\ell(f_{\boldsymbol{\theta}^\ast}) - \hat{R}_\ell(f_{\boldsymbol{\theta}}) \nonumber\\
		&= \mathbb{E}_{\mathbf{x}, y_{\mathbf{x}}}\left[  \frac{C\varphi_{\mathbf{x}}}{K-1} \right]  + 	\mathbb{E}_{\mathbf{x}, y_{\mathbf{x}}} \left[ \left(1 - \frac{K \varphi_{\mathbf{x}}}{K-1}\right) \ell(f_{\boldsymbol{\theta}^\ast}\left( \mathbf{x}\right), \mathbf{e}_{y_{\mathbf{x}}})\right]\nonumber\\
		&- \mathbb{E}_{\mathbf{x}, y_{\mathbf{x}}}\left[  \frac{C\varphi_{\mathbf{x}}}{K-1} \right]  + 	\mathbb{E}_{\mathbf{x}, y_{\mathbf{x}}} \left[ \left(1 - \frac{K \varphi_{\mathbf{x}}}{K-1}\right) \ell(f_{\boldsymbol{\theta}}\left( \mathbf{x}\right), \mathbf{e}_{y_{\mathbf{x}}})\right]\nonumber\\
		&= \mathbb{E}_{\mathbf{x}, y_{\mathbf{x}}} \left[\left( 1 - \frac{K \varphi_{\mathbf{x}}}{K-1}\right)  \left( \ell\left( f_{\boldsymbol{\theta}^\ast}(\mathbf{x}), y_{\mathbf{x}}\right)  -\ell\left( f_{\boldsymbol{\theta}}\left( \mathbf{x}\right), \mathbf{e}_{y_{\mathbf{x}}}\right) \right) \right]\nonumber\\
		&=\mathbb{E}_{\mathbf{x}, y_{\mathbf{x}}} \left[\left( 1 - \frac{K \varphi_{\mathbf{x}}}{K-1}\right)  \left( 0 -\ell\left( f_{\boldsymbol{\theta}}\left( \mathbf{x}\right), \mathbf{e}_{y_{\mathbf{x}}}\right) \right) \right].
	\end{align}
	Since $\varphi_{\mathbf{x}} < \left( {K-1}\right) /{K}$ for all $\mathbf{x} \in \mathcal{X}$, it follows that
	\begin{equation}
		\left( 1 - \frac{K \varphi_{\mathbf{x}}}{K-1}\right)  > 0, \quad \forall \mathbf{x} \in \mathcal{X}.
	\end{equation}
	Together with the nonnegativity in \eqref{eq:geq_0}, this gives
	\begin{equation}
		\hat{R}_\ell(f_{\boldsymbol{\theta}}^\ast) - \hat{R}_\ell(f_{\boldsymbol{\theta}}) \leq 0, \quad \forall f_{\boldsymbol{\theta}}.
	\end{equation}
	Thus, $f_{\boldsymbol{\theta}^\ast}$ is also an optimal solution of the incorrect-label risk $\hat{R}_\ell(f_{\boldsymbol{\theta}})$. Consequently, the robustness condition in Definition~1 is satisfied, and the symmetric loss function $\ell$ is robust to this type of incorrect labels under the stated conditions. 
	
	This completes the proof. $\hfill\blacksquare$

	\section{Proof of Theorem 3}
	\begin{theorem}
		Let $\ell$ be a symmetric loss function. Under the class-dependent error model (Type~3 error), if, for any class $k$ and any sample $\mathbf{x} \in \mathcal{X}$,
		\begin{equation}
			0 \leq \ell\big(f_{\boldsymbol{\theta}}(\mathbf{x}), \mathbf{e}_k\big) \leq \frac{C}{K-1},
		\end{equation}
		Further suppose that $R_\ell(f_{\boldsymbol{\theta}^\ast}) = 0$ and that, 
		for any true class $y \in \mathcal{Y}$ and any incorrect class 
		$j \neq y$,
		\begin{equation}
			\varphi_{y_{\mathbf{x}},j} < 1 - \varphi_{y_{\mathbf{x}}} = 1- \sum_{j \neq 	y_{\mathbf{x}}} {\varphi}_{y_{\mathbf{x}}, j},
		\end{equation}
		then $\ell$ is robust to data with this type of incorrect labels.
	\end{theorem}
	
	\textbf{Proof:} Under the class-dependent error model, for any model $f_{\boldsymbol{\theta}}$, the corresponding population risk $\hat{R}_\ell(f_{\boldsymbol{\theta}})$ can be expanded by the law of total expectation. Using the symmetry property of the loss function, we obtain
	\begin{align}
		&\hat{R}_\ell(f_{\boldsymbol{\theta}})=  \mathbb{E}_{\mathbf{x},y_{\mathbf{x}}} \left[ \mathbb{E}_{\hat{y}_{\mathbf{x}} | \mathbf{x}, y_{\mathbf{x}}}\left[  \ell(f_{\boldsymbol{\theta}}\left( \mathbf{x}\right), \mathbf{e}_{\hat{y}_{\mathbf{x}}}) \right]  \right] \nonumber \\
		&= \mathbb{E}_{\mathbf{x}, y_{\mathbf{x}}} \left[ (1 - \varphi_{y_{\mathbf{x}}}) \ell(f_{\boldsymbol{\theta}}\left( \mathbf{x}\right), \mathbf{e}_{y_{\mathbf{x}}})  +  \sum_{j \neq y_{\mathbf{x}}} {\varphi}_{y_{\mathbf{x}}, j} \ell(f_{\boldsymbol{\theta}}(x), \mathbf{e}_j)\right] \nonumber \\
		&= \mathbb{E}_{\mathbf{x}, y_{\mathbf{x}}}\left[  (1 - \varphi_{y_{\mathbf{x}}}) \left( C - \sum_{j \neq y_{\mathbf{x}}} \ell\left( f_{\boldsymbol{\theta}}\left( \mathbf{x}\right) , \mathbf{e}_j\right) \right) \right]  \label{eq:class_conditional_rewrite}\\
		& \qquad + \mathbb{E}_{\mathbf{x}, y_{\mathbf{x}}} \left[ \sum_{j \neq y_{\mathbf{x}}} {\varphi}_{y_{\mathbf{x}}, j} \ell(f_{\boldsymbol{\theta}}\left( \mathbf{x}\right), \mathbf{e}_j)\right]\nonumber \\
		&= C \mathbb{E}_{\mathbf{x}, y_{\mathbf{x}}} \left[ 1 - \varphi_{y_{\mathbf{x}}}\right]  - \mathbb{E}_{\mathbf{x}, y_{\mathbf{x}}}\left[  \sum_{j \neq y_{\mathbf{x}}} \xi_{y_{\mathbf{x}}, j} \ell(f_{\boldsymbol{\theta}}\left( \mathbf{x}\right), \mathbf{e}_j)\right],\nonumber 
	\end{align}
	where $\xi_{y_{\mathbf{x}}, j} = 1 - \varphi_{y_{\mathbf{x}}} - {\varphi}_{y_{\mathbf{x}}, j}$. By the condition $\varphi_{y_\mathbf{x},j} < 1 - \varphi_{y_\mathbf{x}}$, we have
	\begin{equation}
		\xi_{y_\mathbf{x},j} > 0, \quad \forall j \neq y_\mathbf{x},\ \forall \mathbf{x} \in \mathcal{X}.
	\end{equation}
	
	Next, since $R_\ell(f_{\boldsymbol{\theta}^\ast}) = 0$ and the loss function is nonnegative, it follows that, for correctly labeled sample-label pairs $(\mathbf{x}, y_{\mathbf{x}})$,
	\begin{equation}
		\ell(f_{\boldsymbol{\theta}^\ast}(\mathbf{x}), \mathbf{e}_{y_{\mathbf{x}}}) = 0,
		\quad \mathrm{a.s.};(\mathbf{x},y_{\mathbf{x}})\sim\mathcal{P}.
	\end{equation}
	Combining this with the symmetry condition of $\ell$ yields
	\begin{equation}
		\sum_{j\neq y_{\mathbf{x}}}\ell\left(f_{\boldsymbol{\theta}^\ast}(\mathbf{x}),\mathbf{e}_j\right)=C,\ \forall \mathbf{x} \in \mathcal{X},\ \forall j \neq y_\mathbf{x}.
		\label{eq:totoal_C}
	\end{equation}
	Furthermore, by the boundedness condition in the theorem,
	\begin{equation}
		0 \leq \ell(f_{\boldsymbol{\theta}^\ast}(\mathbf{x}), \mathbf{e}_j) \leq {C}/\left( {K-1}\right) ,\ 	\forall \mathbf{x} \in \mathcal{X},\ \forall j \neq y_\mathbf{x}.
	\end{equation}
	In \eqref{eq:totoal_C}, the sum of the $K-1$ nonnegative terms is $C$, and each term is upper bounded by $C/(K-1)$. Hence, we must have
	\begin{equation} 
		\ell(f_{\boldsymbol{\theta}^\ast}(\mathbf{x}), \mathbf{e}_j) = \frac{C}{K-1}, \quad \forall \mathbf{x} \in \mathcal{X},\ \forall j \neq y_\mathbf{x}.
	\end{equation}
	Substituting this result into \eqref{eq:class_conditional_rewrite} gives
	\begin{equation}
		\hspace{-0.1cm}
		\hat{R}_\ell(f_{\boldsymbol{\theta}^\ast}) = C\mathbb{E}_{\mathbf{x}, y_{\mathbf{x}}}\!\left[ 1 - \varphi_{y_{\mathbf{x}}}\right]  - \mathbb{E}_{\mathbf{x}, y_{\mathbf{x}}}\!\!\left[\sum_{j \neq y_{\mathbf{x}}}\!  \xi_{y_{\mathbf{x}}, j} \frac{C}{K\!-1}\right]\!.
	\end{equation}
	Therefore, for any $f_{\boldsymbol{\theta}}$,
	\begin{equation}
		\hspace{-0.2cm}
		\hat{R}_\ell(f_{\boldsymbol{\theta}^\ast}\!) \!-\!\! \hat{R}_\ell(f_{\boldsymbol{\theta}}) \!\!=\! 	\mathbb{E}_{\mathbf{x}, y_{\mathbf{x}}}\!\!\!\left[\!  \sum_{j \neq y_{\mathbf{x}}} \!\!\xi_{y_{\mathbf{x}}, j} \!\left( \!\ell\!\left(f_{\boldsymbol{\theta}}\!\left( \mathbf{x}\right)\!, \!\mathbf{e}_j\right) \!-\! \frac{C}{K \!\!-\!1}\right) \!\right]\!\!.\!\!
	\end{equation}
	Since $\xi{y_{\mathbf{x}},j}>0$ and
	\begin{equation}
		\ell(f_{\boldsymbol{\theta}}\left( \mathbf{x}\right), \mathbf{e}_j) \leq C/(K-1), \quad \forall \mathbf{x} \in \mathcal{X},\ \forall j \neq y_\mathbf{x}.
	\end{equation}
	we obtain
	\begin{equation}
		\hat{R}_\ell(f_{\boldsymbol{\theta}}^\ast) - \hat{R}_\ell(f_{\boldsymbol{\theta}}) \leq 0,\quad \forall f_{\boldsymbol{\theta}}.
	\end{equation}
	Thus, $f_{\boldsymbol{\theta}^\ast}$ is also an optimal solution of the noisy-label risk $\hat{R}_\ell(f_{\boldsymbol{\theta}})$. Consequently, the symmetric loss function $\ell$ is robust to class-dependent incorrect labels under the stated conditions.
	
	Similarly, following the analysis of the boundary case in the proof of Theorem~1, the boundary condition $\varphi_{y_\mathbf{x},j} = 1 - \varphi_{y_\mathbf{x}}$ cannot, in general, guarantee robustness of $\ell$ to incorrect labels.
	
	This completes the proof. $\hfill\blacksquare$
	
	\section{Proof of logarithmic distribution}
	
	Let $U\sim \mathrm{Unif}(0,1)$ be a random variable, and define
	\begin{equation}
		G=2^U-1.
	\end{equation}
	Then, the random variable $G$ follows a continuous logarithmic distribution on the interval $\left[ 0,1\right]$.
	
	\textbf{Proof:}
	Since the transformation $u\mapsto 2^u-1$ is strictly increasing on $\left[ 0,1\right]$, the support of $G$ is $\left[ 0,1\right]$. The inverse transformation is
	\begin{equation}
		u=\log_2(1+g),\qquad 0\le g\le 1,
	\end{equation}
	with derivative
	\begin{equation}
		\frac{du}{dg}=\frac{1}{(1+g)\ln 2}.
	\end{equation}
	By the monotone transformation formula for random variables, the probability density function of $G$ is
	\begin{equation}
		f_G(g)
		\!=\! f_U\bigl(\log_2(1\!+\!g)\bigr)\!\left|\frac{du}{dg}\right|\!=\! 1\!\cdot\! \frac{1}{(1\!+\!g)\ln 2},
		\  0\!<\!g<\!1\!.
	\end{equation}
	Therefore, $G$ follows a continuous logarithmic distribution on $\left[ 0,1\right]$.
	
	This completes the proof. $\hfill\blacksquare$
	
}

\end{document}